\documentclass[12pt,onecolumn,journal,draftclsnofoot]{IEEEtran}
\usepackage{epsf,psfrag,amssymb,amsfonts,amsmath,graphicx,times,color}
\usepackage[square, comma, sort&compress]{natbib}

\def\boxit#1{\vbox{\hrule\hbox{\vrule\kern3pt
        \vbox{\kern3pt#1\kern3pt}\kern3pt\vrule}\hrule}}

\def\reals{ { {\rm  I \kern-0.15em R }  } }
\def\complex{ {\,{{\rm C} \kern-0.50em \raise0.20ex {  |}}\, }}

\def\Rbf{{\bf R}}

\def\Sc{{\cal S}}

\def\be{\begin{equation}}
\def\ee{\end{equation}}

\def\scalefig#1{\epsfxsize #1\textwidth}
\def\Rxx{\Rbf_{\ssstyle X\kern-.1em X}}

\let\ssstyle=\scriptscriptstyle

\def\ie{{\it i.e.,\ \/}}
\def\Kout{\setbox1=\hbox{\Huge\bf K}\hbox to
1.05\wd1{\hspace{.05\wd1}
}}
\def\Sout{\setbox1=\hbox{\Huge\bf S}\hbox to 1.05\wd1{\hspace{.05\wd1}
}}

\def\ie{{\it i.e.,\ \/}}

\newtheorem{proposition}{Proposition}
\newtheorem{fact}{Fact}
\newtheorem{property}{Property}
\newtheorem{theorem}{Theorem}
\def\scalefig#1{\epsfxsize #1\textwidth}
\def\nn{{\nonumber}}

\begin{document}
\setcounter{page}{0}

\title{Power Control in Cognitive Radio Networks: How to Cross a Multi-Lane Highway\thanks{This work was supported in part by the Army
Research Laboratory CTA on Communication and Networks under Grant
DAAD19-01-2-0011 and by the National Science Foundation under Grants
ECS-0622200 and CCF-0830685. Part of this work was presented at {\it
ITA}, January 2008, and {\it ICASSP}, April 2008.}}
\author{Wei Ren, ~~~ Qing Zhao$^*$, ~~~ Ananthram Swami\thanks{W. Ren and Q. Zhao are with the
Department of Electrical and Computer Engineering, University of
California, Davis, CA 95616. A. Swami is with the Army Research
Laboratory, Adelphi, MD 20783.}
\thanks{$^*$ Corresponding author. Phone: 1-530-752-7390. Fax: 1-530-752-8428. Email: qzhao@ece.ucdavis.edu}}

\maketitle%
\thispagestyle{empty}

\begin{abstract}
We consider power control in cognitive radio networks where
secondary users identify and exploit instantaneous and local
spectrum opportunities without causing unacceptable interference to
primary users. We qualitatively characterize the impact of the
transmission power of secondary users on the occurrence of spectrum
opportunities and the reliability of opportunity detection. Based on
a Poisson model of the primary network, we quantify these impacts by
showing that (i) the probability of spectrum opportunity decreases
exponentially with respect to the transmission power of secondary
users, where the exponential decay constant is given by the traffic
load of primary users; (ii) reliable opportunity detection is
achieved in the two extreme regimes in terms of the ratio between
the transmission power of secondary users and that of primary users.
Such analytical characterizations allow us to study power control
for optimal transport throughput under constraints on the
interference to primary users. Furthermore, we reveal the difference
between detecting primary signals and detecting spectrum
opportunities, and demonstrate the complex relationship between
physical layer spectrum sensing and MAC layer throughput. The
dependency of this PHY-MAC interaction on the application type and
the use of handshake signaling such as RTS/CTS is illustrated.
\end{abstract}


\newpage

\section{Introduction}
\label{sec:intro} \vspace{-0.5em}

Cognitive radio (CR) for opportunistic spectrum access (OSA)
addresses critical challenges in spectrum efficiency, interference
management, and coexistence of heterogeneous networks in future
generations of wireless systems~\citep{Mitola&Maguire:99}. The basic
idea of OSA is to achieve spectrum efficiency and interoperability
through a hierarchical access structure with primary and secondary
users. Secondary users, equipped with cognitive radios capable of
sensing and learning the communication environment, can identify and
exploit instantaneous and local spectrum opportunities without
causing unacceptable interference to primary
users~\citep{Zhao&Sadler:07SPM}.

While conceptually simple, CR for OSA presents new challenges in
every aspect of the system design. In this paper, we focus on
transmission power control. We show that unique features of CR
systems give a fresh twist to this classic problem and call for a
new set of fundamental theories and practical insights for optimal
design. \vspace{-1em}

\subsection{Power Control in Cognitive Radio Networks}
\vspace{-0.5em}

In wireless networks, transmission power defines network topology
and determines network capacity. The tradeoff between long-distance
direct transmission and multi-hop relaying, both in terms of energy
efficiency and network capacity, is now well understood in
conventional wireless networks~\citep{Haenggi&Puccinelli:05COMMag,
Gupta&Kumar:00IT}.

This tradeoff in CR systems is, however, much more complex. The
intricacies of power control in CR systems may be illustrated with
an analogy of crossing a multi-lane highway, each lane having
different traffic load. The objective is to cross the highway as
fast as possible subject to a risk constraint. Should we wait until
all lanes are clear and dash through, or cross one lane at a time
whenever an opportunity arises? What if our ability to detect
traffic in multiple lanes varies with the number of lanes in
question?

We show in this paper that similar questions arise in power control
for secondary users. The transmission power of a secondary user not
only determines its communication range but also affects how often
it sees spectrum opportunities. If a secondary user is to use a high
power to reach its intended receiver directly, it must wait for the
opportunity that no primary receiver is active within its relatively
large interference region, which happens less often. If, on the
other hand, it uses low power, it must rely on multi-hop relaying,
and each hop must wait for its own opportunities to emerge.

A less obvious implication of the transmission power in CR networks
is its impact on the reliability of opportunity detection. As shown
in this paper, the transmission power of a secondary user affects
the performance of its opportunity detector in terms of missed
spectrum opportunities and collisions with primary users. Optimal
power control in CR systems thus requires a careful analysis of the
impact of the transmission power on both the occurrence of
opportunities and the reliability of opportunity detection.
\vspace{-1em}

\subsection{Contributions}
\vspace{-0.5em}

The key contribution of this paper lies in the characterization of
the impact of secondary users' transmission power on the occurrence
of spectrum opportunities and the reliability of opportunity
detection. These impacts of secondary users' transmission power lead
to unique design tradeoffs in CR systems that are nonexistent in
conventional wireless networks and have not been recognized in the
literature of cognitive radio. The recognition and characterization
of these tradeoffs contribute to the fundamental understanding of CR
systems and clarify two major misconceptions in the CR literature,
namely, that the presence/absence of spectrum opportunities is
solely determined by primary \emph{transmitters} and that detecting
primary signals is equivalent to detecting spectrum opportunities.
We show in this paper the crucial role of primary receivers in the
definition of spectrum opportunity, which results in the dependency
of the occurrence of spectrum opportunities on the transmission
power of \emph{secondary users}, in addition to the well understood
dependency on the transmission power of primary users. Furthermore,
we show that spectrum opportunity detection is subject to error even
when primary signals can be perfectly detected. Such an
non-equivalence between detecting primary signals and detecting
spectrum opportunities is the root for the connection between the
reliability of opportunity detection and the transmission power of
secondary users, a connection that has eluded the CR research
community thus far.

The above qualitative and conceptual findings are generally
applicable to various primary and secondary network architectures,
different traffic, signal propagation, and interference models. To
quantify the impact of transmission power on the occurrence of
opportunities and the reliability of opportunity detection, we adopt
a Poisson model of the primary network and a disk model for signal
propagation and interference. Closed-form expressions for the
probability of opportunity and the performance of opportunity
detection (measured by the probabilities of false alarm and miss
detection) are obtained. These closed-form expressions allow us to
establish the exponential decay of the probability of opportunity
with respect to the transmission power and the asymptotical behavior
of the performance of opportunity detection. Specifically, we show
that the probability of opportunity decreases exponentially with
$p_{tx}^{2/\alpha}$, where $p_{tx}$ is the transmission power of
secondary users and $\alpha$ is the path-loss exponent. In terms of
the impact of $p_{tx}$ on spectrum sensing, we show that reliable
opportunity detection is achieved in the {\em two extreme regimes of
the ratio between the transmission power $p_{tx}$ of secondary users
and the transmission power $P_{tx}$ of primary users:
$\frac{p_{tx}}{P_{tx}}\rightarrow 0$ and
$\frac{p_{tx}}{P_{tx}}\rightarrow\infty$}. These quantitative
characterizations lead to a systematic study of optimal power
control in CR systems. Adopting the performance measure of {\em
transport throughput}, we examine how a secondary user should choose
its transmission power according to the interference constraint, the
traffic load and transmission power of primary users, and its own
application type (guaranteed delivery vs. best-effort delivery).

While the disk propagation and interference model is simplistic, it
leads to clean tractable solutions that highlight the main message
regarding the dependencies of the definition, the occurrence, and
the detection of spectrum opportunities on the transmission power of
secondary users. It is our hope that this paper provides insights
for characterizing such dependencies under more complex and more
realistic network and interference models.

Other interesting findings include the difference between detecting
primary signals and detecting spectrum opportunities and how it
affects the performance of spectrum sensing. Furthermore, we
demonstrate the complex relationship between physical layer spectrum
sensing and MAC layer throughput. The dependence of this PHY-MAC
interaction on the application type and the use of handshake
signaling such as RTS/CTS are illustrated. %
\vspace{-1em}

\subsection{Related Work}
\vspace{-0.5em}

The main objectives of power control in conventional wireless
networks are to improve the energy efficiency by appropriately
reducing the transmission power without degrading the link
throughput and/or to increase the total throughput of the network by
enhancing the spatial reuse of the channel~\citep{Krunz&Etal:04INet,
Monks&Etal:00ILCN}. Under these objectives, the tradeoff in power
control between fewer hops and spatial reuse is demonstrated
in~\citep{Haenggi&Puccinelli:05COMMag, Gupta&Kumar:00IT}, and the
impact of transmission power on network performance, such as delay,
connectivity, network throughput, is summarized
in~\citep{Krunz&Etal:04INet, Pantazis&Vergados:07ICST}.

Power control in CR systems has been studied under different network
setups and various performance measures. The design tradeoffs in
power control in conventional wireless networks can still be found
in CR systems. The unique design tradeoffs in power control in CR
systems, \ie the impact of transmission power on the occurrence of
opportunities and the reliability of opportunity detection, however,
has not been recognized or analytically characterized in the
literature.

In~\citep{Srinivasa&Jafar:07GLOCOM, Chen&Etal:08JWC}, power control
for one pair of secondary users coexisting with one pair of primary
users is considered. The use of soft sensing information for optimal
power control is explored in~\citep{Srinivasa&Jafar:07GLOCOM} to
maximize the capacity/SNR of the secondary user under a peak power
constraint at the secondary transmitter and an average interference
constraint at the primary receiver. In~\citep{Chen&Etal:08JWC}, the
secondary transmitter adjusts its transmission power to maximize its
data rate without increasing the outage probability at the primary
receiver. It is assumed in~\citep{Chen&Etal:08JWC} that the channel
gain between the primary transmitter and its receiver is known to
the secondary user. In~\citep{Gao&etal:07GLOCOM}, a power control
strategy based on dynamic programming is developed for one pair of
secondary users under a Markov model of primary users' spectrum
usage.

Power control for OSA in TV bands is investigated
in~\citep{Qian&etal:07LANMAN, Islam&Etal:08JWC}, where the primary
users (TV broadcast) transmit all the time and spatial (rather than
temporal) spectrum opportunities are exploited by secondary users. %
\vspace{-1em}

\subsection{Organization and Notation}
\vspace{-0.5em}

The rest of this paper is organized as follows. Sec.~\ref{sec:SODAD}
provides a qualitative characterization of the impact of
transmission power in CR systems, and lays out the conceptual
foundation for subsequent sections. In Sec.~\ref{sec:Imp_Tx_Pow},
closed-form expressions and properties of the probability of
spectrum opportunity and the performance of opportunity detection
are obtained as functions of the transmission power $p_{tx}$ based
on a Poisson primary network model. Based on these analytical
results, we study power control for optimal transport throughput in
Sec.~\ref{sec:power_trans_throughput}. The impact of RTS/CTS
handshake signaling on the performance of opportunity detection and
the optimal transmission power is examined in
Sec.~\ref{sec:RTS_CTS_LBT}. Sec.~\ref{sec:conclu} concludes the
paper, which also includes some simulation results on the energy
detector and some comments about the effect of multiple secondary
users.

Throughout the paper, we use capital letters for parameters of
primary users and lowercase letters for secondary users. %
\vspace{-1em}

\section{Impact of Transmission Power: Qualitative Characterization}
\label{sec:SODAD} \vspace{-0.5em}

This section lays out the conceptual foundation for subsequent
sections. The impact of transmission power on the occurrence of
opportunities is revealed through a careful examination of the
definitions of spectrum opportunity and interference constraint. The
impact of transmission power on the reliability of opportunity
detection is demonstrated by illuminating the difference between
detecting primary signals and detecting spectrum opportunities. %
\vspace{-1em}

\subsection{Impact on the Occurrence of Spectrum Opportunity}
\label{subsec:Def_SO} \vspace{-0.5em}

A formal investigation of CR systems must start from a clear
definition of spectrum opportunity and interference constraint. To
protect primary users, an interference constraint should specify at
least two parameters $\{\eta,~\zeta\}$. The first parameter $\eta$
is the maximum allowable interference power perceived by an active
primary receiver; it specifies the noise floor and is inherent to
the definition of spectrum opportunity as shown below. The second
parameter $\zeta$ is the maximum outage probability that the
interference at an active primary receiver exceeds the noise floor.
Allowing a positive outage probability $\zeta$ is necessary due to
opportunity detection errors. This parameter is crucial to secondary
users in making transmission decisions based on imperfect spectrum
sensing as shown in~\citep{Chen&etal:07ITsub}.

Spectrum opportunity is a local concept defined with respect to a
particular secondary transmitter and its receiver. Intuitively, {\em
a channel is an opportunity to a pair of secondary users if they can
communicate successfully without violating the interference
constraint\footnote{Here we use channel in a general sense: it
represents a signal dimension (time, frequency, code, etc.) that can
be allocated to a particular user.}.} In other words, the existence
of a spectrum opportunity is determined by two logical conditions:
the reception at the secondary receiver being successful and the
transmission from the secondary transmitter being ``harmless''.
Deceptively simple, this definition has significant implications for
CR systems where primary and secondary users are geographically
distributed and wireless transmissions are subject to path loss and
fading.

For a simple illustration, consider a pair of secondary users ($A$
and $B$) seeking to communicate in the presence of primary users as
shown in Fig.~\ref{fig:SO}. A channel is an opportunity to $A$ and
$B$ if the transmission from $A$ does not interfere with nearby {\it
primary receivers} in the solid circle, and the reception at $B$ is
not affected by nearby {\it primary transmitters} in the dashed
circle. The radius $r_I$ of the solid circle at $A$, referred to as
the interference range of the secondary user, depends on the
transmission power of $A$ and the parameter $\eta$ of the
interference constraint, whereas the radius $R_I$ of the dashed
circle (the interference range of primary users) depends on the
transmission power of primary users and the interference tolerance
of $B$.

The use of a circle to illustrate the interference region is
immaterial. This definition applies to a general signal propagation
and interference model by replacing the solid and dashed circles
with interference footprints specifying, respectively, the subset of
primary receivers who are potential victims of $A$'s transmission
and the subset of primary transmitters who can interfere with
reception at $B$. The key message is that spectrum opportunities
depend on both transmitting and receiving activities of primary
users. Spectrum opportunity is thus a function of (i) the
transmission powers of both primary and secondary users, (ii) the
geographical locations of these users, and (iii) the interference
constraint. Notice also that spectrum opportunities are asymmetric.
A channel that is an opportunity when $A$ is the transmitter and $B$
the receiver may not be an opportunity when $B$ is the transmitter
and $A$ the receiver. This asymmetry leads to a complex dependency
of the optimal transmission power on the application type and the
use of MAC handshake signaling such as RTS/CTS as shown in
Sec.~\ref{sec:power_trans_throughput} and
Sec.~\ref{sec:RTS_CTS_LBT}.

It is clear from the definition of spectrum opportunity that a
higher transmission power (larger $r_I$ in Fig.~\ref{fig:SO}) of the
secondary user requires the absence of active primary receivers over
a larger area, which occurs less often. The impact of transmission
power on the occurrence of opportunity thus follows directly.
\vspace{-1em}

\subsection{Impact on the Performance of Opportunity Detection}
\label{subsec:Det_Oppor} \vspace{-0.5em}

Spectrum opportunity detection can be considered as a binary
hypothesis test. We adopt here the disk signal propagation and
interference model as illustrated in Fig.~\ref{fig:SO}. The basic
concepts presented here, however, apply to a general model.

Let $\mathbb{I}(A,d,\textrm{rx})$ denote the logical condition that
there exist primary receivers within distance $d$ to the secondary
user $A$. Let $\overline{\mathbb{I}(A,d,\textrm{rx})}$ denote the
complement of $\mathbb{I}(A,d,\textrm{rx})$. The two hypotheses for
opportunity detection are then given by %
\vspace{-0.5em} {\small
\begin{eqnarray*}
&\mathcal{H}_0& : \textrm{ opportunity exists}, \ie
\overline{\mathbb{I}(A,r_I,\textrm{rx})}\cap
\overline{\mathbb{I}(B,R_I,\textrm{tx})}, \\
&\mathcal{H}_1& : \textrm{ no opportunity}, \ie
\mathbb{I}(A,r_I,\textrm{rx})\cup \mathbb{I}(B,R_I,\textrm{tx}),
\end{eqnarray*}}
where $\mathbb{I}(B,R_I,\textrm{tx})$ and
$\overline{\mathbb{I}(B,R_I,\textrm{tx})}$ are similarly defined,
and $R_I$ and $r_I$ are, respectively, the interference range of
primary and secondary users under the disk model. Notice that
$\overline{\mathbb{I}(A,r_I,\textrm{rx})}$ corresponds to the
logical condition on the transmission from $A$ being ``harmless'',
and $\overline{\mathbb{I}(B,R_I,\textrm{tx})}$ the logical condition
on the reception at $B$ being successful.

Detection performance is measured by the probabilities of false
alarm $P_F$ and miss detection $P_{MD}$: $P_F = \Pr\{\textrm{decides
}\mathcal{H}_1~|~\mathcal{H}_0\}$, $P_{MD} = \Pr\{\textrm{decides
}\mathcal{H}_0~|~\mathcal{H}_1\}$. The tradeoff between false alarm
and miss detection is captured by the receiver operating
characteristic (ROC), which gives $P_D=1-P_{MD}$ (probability of
detection or detection power) as a function of $P_F$. In general,
reducing $P_F$ comes at the price of increasing $P_{MD}$ and {\it
vice versa}. Since false alarms lead to overlooked spectrum
opportunities and miss detections are likely to result in collisions
with primary users, the tradeoff between false alarm and miss
detection is crucial in the design of CR
systems~\citep{Chen&etal:07ITsub}.

Without assuming cooperation from primary users, the observations
available to the secondary user for opportunity detection are the
signals emitted from primary {\em transmitters}. This basic approach
to opportunity detection is commonly referred to as
``listen-before-talk'' (LBT). As shown in Fig.~\ref{fig:OD}, $A$
infers the existence of spectrum opportunity from the absence of
primary transmitters within its detection range $r_D$, where $r_D$
can be adjusted by changing, for example, the threshold of an energy
detector. The probabilities of false alarm $P_F$ and miss detection
$P_{MD}$ for LBT are thus given by \vspace{-0.5em} {\small
\begin{eqnarray}
P_F = \Pr\{\mathbb{I}(A,r_D,\textrm{tx})~|~\mathcal{H}_0\},~~~~
P_{MD} =
\Pr\{\overline{\mathbb{I}(A,r_D,\textrm{tx})}~|~\mathcal{H}_1\}.
\end{eqnarray}}
\vspace{-2em}

Uncertainties, however, are inherent to such a scheme even if $A$
listens to primary signals with perfect ears ($\ie$ perfect
detection of primary transmitters within its detection range $r_D$).
Even in the absence of noise and fading, the geographic distribution
and traffic pattern of primary users have significant impact on the
performance of LBT. Specifically, there are three possible sources
of detection errors: hidden transmitters, hidden receivers and
exposed transmitters. A {\em hidden transmitter} is a primary
transmitter that is located within distance $R_I$ to $B$ but outside
the detection range of $A$ (node $X$ in Fig.~\ref{fig:OD}). A {\em
hidden receiver} is a primary receiver that is located within the
interference range $r_I$ of $A$ but its corresponding primary
transmitter is outside the detection range of $A$ (node $Y$ in
Fig.~\ref{fig:OD}). An {\em exposed transmitter} is a primary
transmitter that is located within the detection range of $A$ but
transmits to a primary receiver outside the interference range of
$A$ (node $Z$ in Fig.~\ref{fig:OD}). For the scenarios shown in
Fig.~\ref{fig:OD}, even if $A$ can perfectly detect the presence of
signals from any primary transmitters located within its detection
range $r_D$, the transmission from the exposed transmitter $Z$ is a
source of false alarms whereas the transmission from the hidden
transmitter $X$ and the reception at the hidden receiver $Y$ are
sources of miss detections. As illustrated in Fig.~\ref{fig:ROC},
adjusting the detection range $r_D$ leads to different points on the
ROC. It is obvious from (1) that $P_F$ increases but $P_{MD}$
decreases as $r_D$ increases.

From the definition of spectrum opportunity, when
$\frac{p_{tx}}{P_{tx}}$ is small (small $\frac{r_I}{R_I}$ in
Fig.~\ref{fig:SO}), the occurrence of spectrum opportunity is mainly
determined by the logical condition on the reception at $B$ being
successful. In this case, the error in detecting opportunity is
mainly caused by hidden transmitters. Since the distance between $A$
and $B$ is relatively small due to the small transmission power, $A$
can accurately infer the presence of primary transmitters in the
neighborhood of $B$, leading to reliable opportunity detection. On
the other hand, when $\frac{p_{tx}}{P_{tx}}$ is large (large
$\frac{r_I}{R_I}$ in Fig.~\ref{fig:SO}), the occurrence of spectrum
opportunity is mainly determined by the logical condition on the
transmission from $A$ being ``harmless''. Due to the relatively
small transmission power of the primary users, primary receivers are
close to their corresponding transmitters. Node $A$ can thus
accurately infer the presence of primary receivers from the presence
of primary transmitters and achieve reliable opportunity detection.
To summarize, in the two extreme regimes in terms of the ratio
between the transmission power of secondary users and that of
primary users, the two logical conditions for spectrum opportunity
reduce to one. As a consequence, reducing $P_F$ does not necessarily
increase $P_{MD}$, and perfect spectrum opportunity detection is
achieved. A detailed proof of this statement is given in
Sec.~\ref{subsec:Imp_Tx_Pow_Det}. Note that we focus on detection
errors caused by the inherent uncertainties associated with
detecting spectrum opportunities by detecting primary transmitters.
Such uncertainties vary with the transmission power of the primary
and secondary users. We ignore noise and fading that may cause
errors in detecting primary transmitters, since they are not
pertinent to the issue of power control for secondary users.
\vspace{-1em}

\section{Impact of Transmission Power: Quantitative Characterization}
\label{sec:Imp_Tx_Pow} \vspace{-0.5em}

In this section, we quantitatively characterize the impact of the
transmission power $p_{tx}$ of secondary users by deriving
closed-form expressions for the probability of opportunity and the
performance of opportunity detection as functions of $p_{tx}$ in a
Poisson primary network. The exponential decay of the probability of
opportunity with respect to $p_{tx}^{2/\alpha}$ and the asymptotic
behavior of ROC are established based on these closed-form
expressions. \vspace{-1em}

\subsection{A Poisson Random Network Model}
\label{subsec:Poisson_Model} \vspace{-0.5em}

Consider a decentralized primary network, where potential primary
transmitters are distributed according to a two-dimensional
homogeneous Poisson process with density $\lambda$. At the beginning
of each slot, each potential primary transmitter has a probability
$p$ to transmit, and receivers are in turn located uniformly within
the transmission range $R_p$ of each transmitter. In the following
analysis, we will frequently use the following two classic results
on Poisson processes.
\begin{fact}{\it Independent Coloring (Thinning)
Theorem}~[14, Chapter 5] \\
Let $\Pi$ be a potentially inhomogeneous Poisson process on
$\mathbb{R}^d$ with density function $\lambda(\mathbf{x})$, where
$\mathbf{x}=(x_1,x_2,...,x_d)\in \mathbb{R}^d$. Suppose that we
obtain $\Pi'$ by independently coloring points $\mathbf{x}\in \Pi$
according to probabilities $p(\mathbf{x})$. Then $\Pi'$ and
$\Pi-\Pi'$ are two independent Poisson processes with density
functions $p(\mathbf{x})\lambda(\mathbf{x})$ and
$(1-p(\mathbf{x}))\lambda(\mathbf{x})$, respectively.
\end{fact}

\begin{fact}{\it Displacement Theorem}~[14, Chapter 5] \\
Let $\Pi$ be a Poisson process on $\mathbb{R}^d$ with density
function $\lambda(\mathbf{x})$. Suppose that the points of $\Pi$ are
displaced randomly and independently. Let
$\rho(\mathbf{x},\mathbf{y})$ denote the probability density of the
displaced position $\mathbf{y}$ of a point $\mathbf{x}$ in $\Pi$.
Then the displaced points form a Poisson process $\Pi'$ with density
function $\lambda'$ given by
$\lambda'(\mathbf{y})=\int_{\mathbb{R}^d}
\lambda(\mathbf{x})\rho(\mathbf{x},\mathbf{y})~\mathrm{d}\mathbf{x}$.
In particular, if $\lambda(\mathbf{x})$ is a constant $\lambda$ and
$\rho(\mathbf{x},\mathbf{y})$ is a function of
$\mathbf{y}-\mathbf{x}$, then $\lambda'(\mathbf{y})=\lambda$ for all
$\mathbf{y}\in \mathbb{R}^d$.
\end{fact}

Note that in the independent coloring theorem, the original Poisson
process $\Pi$ does not have to be homogeneous and the coloring
probability $p(\mathbf{x})$ can depend on the location $\mathbf{x}$.
This theorem is more general than the commonly known thinning
theorem for homogeneous Poisson processes. In our subsequent
analysis, we rely on this general version of the thinning theorem to
handle location-dependent coloring/thinning.

Based on Fact 1 and Fact 2, we arrive at the following property.
\begin{property}{\it Distributions of Primary Transmitters and Primary Receivers} \\
Both primary transmitters and receivers form a homogeneous Poisson
process with density $p\lambda$. %
\end{property}
Note that although the two Poisson processes have the same density,
they are not independent. \vspace{-1em}

\subsection{Impact of Transmission Power on Probability of
Opportunity} \label{subsec:Imp_Tx_Pr_H0}

Let $d$ be the distance between $A$ and $B$. Let $\Sc_I(d,r_1,r_2)$
denote the common area of two circles centered at $A$ and $B$ with
radii $r_1$ and $r_2$, respectively, and $\Sc_c(d,r_1,r_2)$ denote
the area within a circle with radius $r_1$ centered at $A$ but
outside the circle with radius $r_2$ centered at $B$ (see
Fig.~\ref{fig:Sc_0}). Then we have the following proposition.
\begin{proposition}{\it Closed-form Expression for Probability of
Opportunity} \\
Under the disk signal propagation and interference model
characterized by $\{r_I,R_I\}$, the probability of opportunity for a
pair of secondary users $A$ and $B$ in a Poisson primary network
with density $\lambda$ and traffic load $p$ is given by {\small
\begin{eqnarray}
\Pr[\mathcal{H}_0] =
\exp\left[-p\lambda\left(\underset{\Sc_c(d,r_I+R_p,R_I)}
{\iint}\frac{\Sc_I(r,R_p,r_I)}{\pi R_p^2}r\mathrm{d}r\mathrm{d}\theta+\pi R_I^2\right)\right], %
\end{eqnarray}}
where the secondary transmitter $A$ is chosen as the origin of the
polar coordinate system for the double integral, and $d$ the
distance between $A$ and $B$.
\end{proposition}

\begin{proof}
Based on the definition of spectrum opportunity, we have
\vspace{-0.5em} {\small
\begin{eqnarray}
\Pr[\mathcal{H}_0] &=&
\Pr\{\overline{\mathbb{I}(A,r_I,\textrm{rx})}\cap
\overline{\mathbb{I}(B,R_I,\textrm{tx})}\}, \nn\\ &=&
\Pr\{\overline{\mathbb{I}(A,r_I,\textrm{rx})}~|~\overline{\mathbb{I}(B,R_I,\textrm{tx})}\}\Pr\{\overline{\mathbb{I}(B,R_I,\textrm{tx})}\}.%
\end{eqnarray}}
\vspace{-2em}

Based on Property 1, the second term of (3) is given by
\vspace{-0.5em}{\small
\begin{eqnarray}
\Pr\{\overline{\mathbb{I}(B,R_I,\textrm{tx})}\} = \exp(-p\lambda \pi
R_I^2).
\end{eqnarray}}
Next we obtain the first term
$\Pr\{\overline{\mathbb{I}(A,r_I,\textrm{rx})}~|~\overline{\mathbb{I}(B,R_I,\textrm{tx})}\}$
of (3) based on Fact 1 with location-dependent coloring.

Let $\Pi_{\textrm{tx}}$ denote the Poisson process formed by primary
transmitters. If we color those primary transmitters in
$\Pi_{\textrm{tx}}$ whose receivers are within distance $r_I$ to
$A$, then from Fact 1 we obtain another Poisson process
$\Pi'_{\textrm{tx}}$ formed by all the colored primary transmitters
with density $p\lambda \frac{\Sc_I(r,R_p,r_I)}{\pi R_p^2}$, where
$\frac{\Sc_I(r,R_p,r_I)}{\pi R_p^2}$ is the coloring probability for
a primary transmitter at distance $r$ to $A$.

Given $\overline{\mathbb{I}(B,R_I,\textrm{tx})}$, $\ie$ there are no
primary transmitters within distance $R_I$ to $B$, those primary
receivers within distance $r_I$ to $A$ can only communicate with
those primary transmitters inside $\Sc_c(d,r_I+R_p,R_I)$. Thus, the
event $\overline{\mathbb{I}(A,r_I,\textrm{rx})}$ (conditioned on
$\overline{\mathbb{I}(B,R_I,\textrm{tx})}$) occurs if and only if
$\Pi'_{\textrm{tx}}$ does not have any points inside
$\Sc_c(d,r_I+R_p,R_I)$, $\ie$ {\small
\begin{eqnarray}
\Pr\{\overline{\mathbb{I}(A,r_I,\textrm{rx})}~|~\overline{\mathbb{I}(B,R_I,\textrm{tx})}\}=\exp\left(-\underset{\Sc_c(d,r_I+R_p,R_I)}
{\iint}p\lambda \frac{\Sc_I(r,R_p,r_I)}{\pi
R_p^2}r\mathrm{d}r\mathrm{d}\theta\right).
\end{eqnarray}}
Then by substituting (4, 5) into (3), we arrive at (2).
\end{proof}

While the closed-form expression for $\Pr[\mathcal{H}_0]$ given in
(2) appears to be complex with a double integral, it has a simple
structure that allows us to establish the monotonicity and the
exponential decay rate of $\Pr[\mathcal{H}_0]$ with respect to
$r_I^2$ as given in Theorem 1. Furthermore, as shown in Appendix A,
by integrating with respect to $\theta$ first, we can reduce the
double integral in (2) to a single integral $\int_0^{r_I+R_p}
\frac{\Sc_I(r,R_p,r_I)}{\pi R_p^2}r\theta(r)~\mathrm{d}r$, where
$\theta(r)$ is a function of the radial coordinate $r$ and is
determined by the shape of $\Sc_c(d,r_I+R_p,R_I)$. The basic idea is
that the integrand in (2) is not a function of the angular
coordinate $\theta$ and the range of $\theta$ as a function of the
radial coordinate $r$ can be obtained in an explicit form. In the
integrand of the obtained single integral, $\Sc_I(r,R_p,r_I)$ that
depends on $r$ is also in an explicit form as obtained
in~\citep{Two_circles} and provided in Appendix A. As a consequence,
the resulting single integral is easy to compute.

From (2), we obtain the following theorem that characterizes the
impact of transmission power $p_{tx}$ on the probability of
opportunity.

\begin{theorem} {\it Impact of Transmission Power on Opportunity Occurrence}
\begin{itemize}
\item[T1.1.] $\Pr[\mathcal{H}_0]$ is a strictly decreasing function of
$p_{tx}\propto r_I^{\alpha}$.
\item[T1.2.] $\Pr[\mathcal{H}_0]$ decreases exponentially\footnote{A quantity $N$ is said to decrease
exponentially with respect to $t$ if its decay rate is proportional
to its value. Symbolically, this can be expressed as the following
differential equation: $\frac{\mathrm{d}N}{\mathrm{d}t}=-\lambda N$,
where $\lambda>0$ is called the decay constant.} with
$p_{tx}^{2/\alpha}\propto r_I^2$, where the decay constant is
proportional to $p\lambda$, $\ie$ $\exp(-p\lambda \pi
R_I^2)<\frac{\Pr[\mathcal{H}_0]}{\exp(-p\lambda \pi r_I^2)}\leq 1$,
with equality when $r_I\geq d+R_I+R_p$.
\item[T1.3.] $\Pr[\mathcal{H}_0]$ decreases exponentially with
$p\lambda$, where the decay constant is $\pi r_I^2\propto
p_{tx}^{2/\alpha}$.
\end{itemize}
\end{theorem}

\begin{proof}
Theorem 1 is obtained by examining the closed-form expression for
$\Pr[\mathcal{H}_0]$ given in (2). Details are given in Appendix B.
\end{proof}

From T1.2, we can see that when the transmission power of secondary
users is high ($r_I\geq d+R_I+R_p$), the probability of opportunity
$\Pr[\mathcal{H}_0]$ has a simple expression:
$\Pr[\mathcal{H}_0]=\exp(-p\lambda \pi r_I^2)$. When $r_I\geq
d+R_I+R_p$, the absence of primary receivers within distance $r_I$
to $A$ automatically implies the absence of primary transmitters
within distance $R_I$ to $B$. Thus, the opportunity occurs if and
only if there is no primary receiver within distance $r_I$ to $A$,
which leads to the simple expression for $\Pr[\mathcal{H}_0]$.
Moreover, from T1.2 we can see that the traffic load $p\lambda$ of
primary users determines the exponential decay rate of
$\Pr[\mathcal{H}_0]$ with respect to $p_{tx}^{2/\alpha}$. Similarly,
T1.3 shows that the area $\pi r_I^2$ ``consumed'' by the secondary
transmitter, a concept introduced in~\citep{Gupta&Kumar:00IT}, is
the decay constant of $\Pr[\mathcal{H}_0]$ with respect to
$p\lambda$.

A numerical example is given in Fig.~\ref{fig:log_Pr_H0}(a), where
$\Pr[\mathcal{H}_0]$ and its lower and upper bounds ($\exp[-p\lambda
\pi(r_I^2+R_I^2)]$ and $\exp(-p\lambda \pi r_I^2)$, respectively)
are plotted as a function of $r_I$. The exponential decay rate of
$\Pr[\mathcal{H}_0]$ can be easily observed by noticing the log
scale. Fig.~\ref{fig:log_Pr_H0}(b) demonstrates $\Pr[\mathcal{H}_0]$
as a function of $p\lambda$ for different $r_I$. It shows that the
exponential decay constant of $\Pr[\mathcal{H}_0]$ with respect to
$p\lambda$ increases as $r_I$ increases. \vspace{-1em}

\subsection{Impact of Transmission Power on Detection Performance}
\label{subsec:Imp_Tx_Pow_Det} \vspace{-0.5em}

In the following, we focus on the performance of the spectrum
opportunity detector for one pair of secondary users $A$ and $B$,
where there are no other secondary users in the network.

For LBT, false alarms occur if and only if there exist primary
transmitters within the detection range $r_D$ of $A$ under
$\mathcal{H}_0$, and miss detections occur if and only if there is
no primary transmitter within the detection range $r_D$ of $A$ under
$\mathcal{H}_1$. We thus have the following proposition.
\begin{proposition}{\it Closed-form Expressions for $P_F$ and $P_{MD}$} \\
Under the disk model characterized by $\{r_I,R_I\}$, let $r_D$ be
the detection range. The probabilities of false alarm $P_F$ and miss
detection $P_{MD}$ for a pair of secondary users $A$ and $B$ in a
Poisson primary network with density $\lambda$ and traffic load $p$
are given by \vspace{-0.5em} {\small
\begin{eqnarray}
P_F &=& 1-\exp\left[-p\lambda\left(\pi
r_D^2-\Sc_I(d,r_D,R_I)-\underset{\Sc_o}{\iint}\frac{\Sc_I(r,R_p,r_I)}{\pi
R_p^2}r\mathrm{d}r\mathrm{d}\theta\right)\right],
\end{eqnarray}} \vspace{-0.5em}
{\small
\begin{eqnarray}
P_{MD} &=& \frac{\exp(-p\lambda \pi r_D^2)-\exp\left[-p\lambda \left(\pi (r_D^2+R_I^2)-\Sc_I(d,r_D,R_I)+\underset{\Sc_c(d,r_I+R_p,R_I)-\Sc_o}{\iint}%
         \frac{\Sc_I(r,R_p,r_I)}{\pi R_p^2}r\mathrm{d}r\mathrm{d}\theta\right)\right]}{1-\exp\left[-p\lambda\left(%
         \underset{\Sc_c(d,r_I+R_p,R_I)}{\iint}\frac{\Sc_I(r,R_p,r_I)}{\pi R_p^2}r\mathrm{d}r\mathrm{d}\theta+\pi
         R_I^2\right)\right]}, ~~~~~~
\end{eqnarray}}
where the secondary transmitter $A$ is chosen as the origin of the
polar coordinate system for the double integral, $d$ the distance
between $A$ and $B$, and $\Sc_o=\Sc_c(d,r_D,R_I)\cap
\Sc_c(d,r_I+R_p,R_I)$.
\end{proposition}

\begin{proof}
Similar to Proposition 1, the proof uses Fact 1 with
location-dependent coloring. For details, see Appendix C.
\end{proof}

Similarly to (2), the double integral in (6, 7) can be simplified to
a single integral due to the independence of the integrand with
respect to the angular coordinate $\theta$ and the special shape of
$\Sc_o$. Due to the page limit, the details are left
in~\citep{RenZhaoTR0705}.

From Proposition 2, we can show the following theorem that
characterizes the impact of the transmission power of secondary
users (represented by $r_I$) on the asymptotic behavior of the ROC
for spectrum opportunity detection.

\begin{theorem} {\em Impact of Transmission Power on Detection
Performance}\footnote{Since the minimum transmission power for
successful reception is in general higher than the maximum allowable
interference power, it follows that the transmission range $R_p$ of
primary users is smaller than $R_I$. Furthermore, under the disk
signal propagation and interference model, we have $R_p=\beta R_I$
($0<\beta<1$). A similar relationship holds for $d$ and $r_I$.}. \\
There exist two points on ROC that asymptotically approach $(0,1)$
as $\frac{r_I}{R_I}\rightarrow 0 \textrm{ and }\infty$,
respectively. Specifically, \vspace{-1em} {\small
\begin{eqnarray*}
& &\underset{\frac{r_I}{R_I}\rightarrow
0}{\lim}(P_F(r_D=R_I),P_D(r_D=R_I))=(0,1), \\
& &\underset{\frac{r_I}{R_I}\rightarrow
\infty}{\lim}(P_F(r_D=r_I-R_I),P_D(r_D=r_I-R_I))=(0,1).
\end{eqnarray*}}
\end{theorem}

\begin{proof}
The intuitive reasons for choosing $r_D=R_I$ and $r_D=r_I-R_I$ in
the two extreme regimes are discussed in
Sec.~\ref{subsec:Det_Oppor}. For details of the proof, see Appendix
D.
\end{proof}

Since $(0,1)$ is the perfect operating point on a ROC, we can
asymptotically approach perfect detection performance by choosing
$r_D=R_I$ when $\frac{r_I}{R_I}\rightarrow 0$ or $r_D=r_I-R_I$ when
$\frac{r_I}{R_I}\rightarrow \infty$. A numerical example is shown in
Fig.~\ref{fig:ROC_extreme_LBT}, where ROC approaches the corner
$(0,1)$ as $\frac{r_I}{R_I}$ increases or decreases. \vspace{-1em}

\section{Power Control for Optimal Transport Throughput}
\label{sec:power_trans_throughput} \vspace{-0.5em}

In this section, the impact of the transmission power on the
occurrence of opportunities and the impact on opportunity detection
are integrated together for optimal power control. Under the
performance measure of transport throughput subject to an
interference constraint, we examine how a secondary user should
choose its transmission power according to the interference
constraint, the traffic load and transmission power of primary
users, and its own application type (guaranteed delivery vs.
best-effort delivery). \vspace{-1em}

\subsection{Transport Throughput}
\vspace{-0.5em}

From Sec.~\ref{subsec:Imp_Tx_Pr_H0} and
Sec.~\ref{subsec:Imp_Tx_Pow_Det}, it seems that the transmission
power $p_{tx}$ should be chosen as small as possible to maximize the
probability of opportunity and improve detection quality. Such a
choice of the transmission power, however, does not lead to an
efficient communication system due to the small distance covered by
the transmission. We adopt here transport throughput as the
performance measure, which is defined as \vspace{-0.5em} %
{\small
\begin{eqnarray}
C(r_I,r_D) = d(r_I)\cdot P_S(r_D,d(r_I)),
\end{eqnarray}}
where $d(\propto r_I)$ is the transmission range of the secondary
user, and $P_S$ is the probability of successful data transmission
which depends on both the occurrence of opportunities and the
reliability of opportunity detection. Then power control for optimal
transport throughput can be formulated as a constrained optimization
problem: \vspace{-0.5em} {\small
\begin{eqnarray}
r^*_I = \underset{r_I}{\arg\max}\{C\} =
\underset{r_I}{\arg\max}\{d(r_I)\cdot P_S(r_D,d(r_I))\} \textrm{
s.t. } P_C(r_D,r_I)\leq \zeta,
\end{eqnarray}}
where $\zeta$ is maximum allowable collision probability, $P_C$ the
probability of colliding with primary users which depends on the
reliability of opportunity detection. Note that the detection range
$r_D$ is not an independent parameter; it is determined by
maximizing $P_S(r_D,d(r_I))$ subject to $P_C(r_D,r_I)\leq \zeta$ for
every given interference range $r_I$.

In order to solve the above constrained optimization problem, we
need expressions for $P_C$ and $P_S$ which collectively measure the
MAC layer performance. \vspace{-1em}

\subsection{MAC Performance of LBT}
\label{subsec:MAC_Per_LBT} \vspace{-0.5em}

We first consider $P_S$, which is application dependent. For
applications requiring guaranteed delivery, an acknowledgement (ACK)
signal from the secondary receiver $B$ to the secondary transmitter
$A$ is required to complete a data transmission. Specifically, in a
successful data transmission, the following three events should
occur in sequence: $A$ detects the opportunity
($\overline{\mathbb{I}(A,r_D,\textrm{tx})}$) and transmits data to
$B$; $B$ receives data successfully
($\overline{\mathbb{I}(B,R_I,\textrm{tx})}$) and replies to $A$ with
an ACK; $A$ receives the ACK
($\overline{\mathbb{I}(A,R_I,\textrm{tx})}\}$) which completes the
transmission. We thus have \vspace{-0.5em} {\small
\begin{eqnarray}
P_S &=& \Pr\{\overline{\mathbb{I}(A,r_D,\textrm{tx})}\cap
\overline{\mathbb{I}(B,R_I,\textrm{tx})}\cap
\overline{\mathbb{I}(A,R_I,\textrm{tx})}\}, \nn\\ &=&
\Pr\{\overline{\mathbb{I}(A,r_E,\textrm{tx})}\cap
\overline{\mathbb{I}(B,R_I,\textrm{tx})}\},
\end{eqnarray}}
where $r_E=\max\{r_D,R_I\}$.

For best-effort delivery applications~\citep{Best_effort},
acknowledgements are not required to confirm the completion of data
transmissions. In this case, we have \vspace{-0.5em} {\small
\begin{eqnarray}
P_S = \Pr\{\overline{\mathbb{I}(A,r_D,\textrm{tx})}\cap
\overline{\mathbb{I}(B,R_I,\textrm{tx})}\}.
\end{eqnarray}}
\vspace{-2em}

The probability of collision is defined as \footnote{In obtaining
the definition (12) of $P_C$, we have assumed that the interference
caused by the ACK signal is negligible due to its short duration.}
\vspace{-0.5em} {\small
\begin{eqnarray}
P_C = \Pr\{\textrm{$A$ transmits
data}~|~\mathbb{I}(A,r_I,\textrm{rx})\}.
\end{eqnarray}}
Note that $P_C$ is conditioned on $\mathbb{I}(A,r_I,\textrm{rx})$
instead of $\mathcal{H}_1$. Clearly,
$\Pr[\mathbb{I}(A,r_I,\textrm{rx})]\le \Pr[\mathcal{H}_1]$.

Since the secondary transmitter $A$ transmits data if and only if
$A$ detects no nearby primary transmitters, we thus have
\vspace{-0.5em} {\small
\begin{eqnarray}
P_C =
\Pr\{\overline{\mathbb{I}(A,r_D,\textrm{tx})}~|~\mathbb{I}(A,r_I,\textrm{rx})\}.
\end{eqnarray}}
\vspace{-2em}

By considering the Poisson primary network and the disk model, we
obtain the closed-form expressions for $P_C$ and $P_S$ given in the
following proposition.
\begin{proposition}{\it Closed-form Expressions for $P_C$ and $P_S$} \\
Under the disk model characterized by $\{r_I,R_I\}$, let $r_D$ be
the detection range. The probabilities of collision $P_C$ and
successful transmission $P_S$ for a pair of secondary users $A$ and
$B$ in a Poisson primary network with density $\lambda$ and traffic
load $p$ are given by {\small
\begin{eqnarray}
P_C &=& \frac{\exp(-p\lambda \pi r_D^2)[1-\exp(-p\lambda
\pi(r_I^2-I(r_D,r_I,R_p)))]}{1-\exp(-p\lambda \pi r_I^2)}, \\
P_S &=& \left\{\begin{array}{ll}
\exp[-p\lambda(\pi(r_E^2+R_I^2)-\Sc_I(d,r_E,R_I))], & \textrm{for guaranteed delivery,} \\
\exp[-p\lambda(\pi(r_D^2+R_I^2)-\Sc_I(d,r_D,R_I))], & \textrm{for best-effort delivery,} %
\end{array}
\right.
\end{eqnarray}}
where $I(r_D,r_I,R_p)=\int_{0}^{r_D} 2r\frac{\Sc_I(r,r_I,R_p)}{\pi
R_p^2}\mathrm{d}r$.
\end{proposition}
\begin{proof}
Similar to Proposition 2, the derivation of $P_C$ uses Fact 1 with
location-dependent coloring. For details, see Appendix E. The
expressions for $P_S$ follow immediately from (10, 11) and Property
1.
\end{proof}

Based on the expression for $\Sc_I(r,r_I,R_p)$, we can obtain
$I(r_D,r_I,R_p)$ in an explicit form without integral. Details are
left in~\citep{RenZhaoTR0705} due to the page limit. Notice that the
above expressions for $P_C$ and $P_S$ are in explicit form without
integrals. With the explicit expressions for $P_C$ (14) and $P_S$
(15), the constrained optimization given in (9) can be solved
numerically. \vspace{-1em}

\subsection{Numerical Examples}
\vspace{-0.5em}

Shown in Fig.~\ref{fig:trans_throughput_p_01_LBT} is a numerical
example where we plot transport throughput $C$ as a function of
$r_I$. Notice that $r_I^*$ is the interference range for optimal
transport throughput. We can see that $r_I^*$ for best-effort
applications is different from that for guaranteed delivery, and
neither of them is not in the two extreme regimes. We can also see
that the optimal transport throughput for best-effort delivery is
larger than that for guaranteed delivery. This example thus
demonstrates that OSA based on cognitive radio is more suitable for
best-effort applications as compared to guaranteed delivery due to
the asymmetry of spectrum opportunities.

Fig.~\ref{fig:opt_r_I_LBT} shows how the optimal transmission power
of the secondary user changes with the traffic load and transmission
power of the primary users, as well as the application type of the
secondary user. Specifically, the optimal interference range $r_I^*$
decreases as the traffic load increases. This agrees with our
intuition from the analogy of crossing a multi-lane highway.
Furthermore, the optimal transmission power of the secondary user is
related to that of the primary user. We can see from Fig 8 that when
the traffic load is low, $r^*_I$ is close to the interference range
$R_I$ of primary transmitters. When the traffic load is high,
$r^*_I$ is much smaller than $R_I$.

When the traffic load is low, $r^*_I$ for both application types are
the same. This is because for both application types, the optimal
detection range $r_D$ corresponding to each $r_I$ in the
neighborhood of $r^*_I$ is larger than $R_I$, which leads to the
same probability of successful transmission $P_S$ (see (15)) and
thus the same maximum point $r^*_I$ for transport throughput $C$.
When the traffic load is high, $r^*_I$ for guaranteed delivery is
smaller than that for best-effort delivery, which is consistent with
the case shown in Fig.~\ref{fig:trans_throughput_p_01_LBT}. This is
because for guaranteed delivery, the optimal detection range $r_D$
corresponding to each $r_I$ in the neighborhood of $r^*_I$ is
smaller than $R_I$, which leads to a smaller $P_S$ than that for
best-effort delivery (see (15)) and thus a smaller maximum point
$r^*_I$. \vspace{-1em}

\section{RTS/CTS-Enhanced LBT}
\label{sec:RTS_CTS_LBT} \vspace{-0.5em}

The sources of the detection error floor of LBT in the absence of
noise and fading resemble the hidden and exposed terminal problem in
the conventional ad hoc networks of peer users. It is thus natural
to consider the use of RTS/CTS handshake signaling to enhance the
detection performance of LBT. For RTS/CTS-enhanced LBT (ELBT),
spectrum opportunity detection is performed jointly by the secondary
transmitter $A$ and the secondary receiver $B$ through the exchange
of RTS/CTS signals. The detailed steps are given below.
\begin{itemize}
\item $A$ detects primary transmitters within distance $r_D$.
If it detects none, $A$ sends $B$ a Ready-to-Send (RTS) signal.
\item If $B$ receives the RTS signal from $A$ successfully,
then $B$ replies with a Clear-to-Send (CTS) signal.
\item Upon receiving the CTS signal, $A$
transmits data to $B$.
\end{itemize}

Since for ELBT, the observation space comprises the RTS and CTS
signals, we have the following for the probabilities of false alarm
$P_F$ and miss detection $P_{MD}$. \vspace{-0.5em} {\small
\begin{eqnarray}
P_F
&=& \Pr\{\textrm{failed RTS/CTS exchange }|~\mathcal{H}_0\}, \nonumber\\
&=& \Pr\{\mathbb{I}(A,r_D,\textrm{tx})\cup \mathbb{I}(B,R_I,\textrm{tx})\cup \mathbb{I}(A,R_I,\textrm{tx})~|~\mathcal{H}_0\}, \nonumber\\
&=& \Pr\{\mathbb{I}(A,r_E,\textrm{tx})~|~\mathcal{H}_0\},
\end{eqnarray}}
where the last step follows from
$\Pr\{\mathbb{I}(B,R_I,\textrm{tx})\cap \mathcal{H}_0\}=0$.
\vspace{-0.5em} {\small
\begin{eqnarray}
P_{MD}
&=& \Pr\{\textrm{successful RTS/CTS exchange }|~\mathcal{H}_1\}, \nonumber\\
&=& \Pr\{\overline{\mathbb{I}(A,r_D,\textrm{tx})}\cap \overline{\mathbb{I}(B,R_I,\textrm{tx})}\cap \overline{\mathbb{I}(A,R_I,\textrm{tx})}~|~\mathcal{H}_1\}, \nonumber\\
&=& \Pr\{\overline{\mathbb{I}(A,r_E,\textrm{tx})}\cap
\overline{\mathbb{I}(B,R_I,\textrm{tx})}~|~\mathcal{H}_1\}.
\end{eqnarray}} \vspace{-2em}

Since $A$ transmits data if and only if an successful RTS/CTS
exchange occurs, it follows that\footnote{In obtaining the
definition (18) of $P_C$, we have assumed that the interference
caused by the RTS, CTS and ACK signals is negligible due to their
short durations.} \vspace{-2em} {\small
\begin{eqnarray}
P_C &=& \Pr\{\overline{\mathbb{I}(A,r_E,\textrm{tx})}\cap
\overline{\mathbb{I}(B,R_I,\textrm{tx})}~|~\mathbb{I}(A,r_I,\textrm{rx})\}.
\end{eqnarray}}
\vspace{-2em}

Unlike LBT, miss detections always lead to successful data
transmissions for ELBT. This is because miss detections can only
occur after a successful RTS-CTS exchange. Then $B$ can receive data
successfully as it can receive RTS. We thus have \vspace{-0.5em}
{\small
\begin{eqnarray}
P_S &=& (1-P_F)\cdot \Pr[\mathcal{H}_0]+P_{MD}\cdot
\Pr[\mathcal{H}_1], \nn\\ &=&
\Pr\{\overline{\mathbb{I}(A,r_E,\textrm{tx})}\cap
\overline{\mathbb{I}(B,R_I,\textrm{tx})}\}.
\end{eqnarray}}
Notice that $P_S$ of ELBT is identical to that of LBT for guaranteed
delivery in (10). Due to the requirement on the successful reception
of CTS in opportunity detection, $P_S$ for ELBT is independent of
the application, $\ie$ whether or not ACK is required.

Based on (16-19), we obtain the following proposition for the
Poisson primary network model. \vspace{-1.5em}
\begin{proposition} {\it Closed-form Expressions for $P_F$, $P_{MD}$, $P_C$, and
$P_S$} \\
Under the disk model characterized by $\{r_I,R_I\}$, let $r_D$ be
the detection range. The probabilities of false alarm $P_F$, miss
detection $P_{MD}$, collision $P_C$ and successful transmission
$P_S$ for a pair of secondary users $A$ and $B$ in a Poisson primary
network with density $\lambda$ and traffic load $p$ are given by
\vspace{-0.5em} {\small
\begin{eqnarray}
P_F = 1-\exp\left[-p\lambda\left(\pi
r_E^2-\Sc_I(d,r_E,R_I)-\underset{\Sc'_o}{\iint}\frac{\Sc_I(r,R_p,r_I)}{\pi
R_p^2}r\mathrm{d}r\mathrm{d}\theta\right)\right],
\end{eqnarray}}
where $\Sc'_o=\Sc_c(d,r_E,R_I)\cap \Sc_c(d,r_I+R_p,R_I)$.
\vspace{-0.5em} {\small
\begin{eqnarray}
P_{MD}
=\frac{\exp[-p\lambda(\pi(r_E^2+R_I^2)-\Sc_I(d,r_E,R_I))]\left[1-\exp\left(-p\lambda\underset{\Sc_c(d,r_I+R_p,R_I)-\Sc'_o}{\iint}\frac{\Sc_I(r,R_p,r_I)}{\pi R_p^2}r\mathrm{d}r\mathrm{d}\theta\right)\right]}%
{1-\exp\left[-p\lambda\left(%
\underset{\Sc_c(d,r_I+R_p,R_I)}{\iint}\frac{\Sc_I(r,R_p,r_I)}{\pi R_p^2}r\mathrm{d}r\mathrm{d}\theta+\pi R_I^2\right)\right]}. %
\end{eqnarray}}
\vspace{-0.5em} {\small
\begin{eqnarray}
P_C
=\frac{\exp[-p\lambda(\pi(r_E^2+R_I^2)-\Sc_I(d,r_E,R_I))]\left[1-\exp\left(-p\lambda\underset{\Sc_c(d,r_I+R_p,R_I)-\Sc'_o}{\iint}\frac{\Sc_I(r,R_p,r_I)}{\pi R_p^2}r\mathrm{d}r\mathrm{d}\theta\right)\right]}%
                {1-\exp\left(p\lambda \pi r_I^2\right)}.
\end{eqnarray}
\vspace{-1em}
\begin{eqnarray}
P_S = \exp[-p\lambda(\pi(r_E^2+R_I^2)-\Sc_I(d,r_E,R_I))].
\end{eqnarray}}
Furthermore, Theorem 2 still holds for ELBT, $\ie$ perfect detection
performance can be achieved at $r_D=R_I$ when
$\frac{r_I}{R_I}\rightarrow 0$ and at $r_D=r_I-R_I$ when
$\frac{r_I}{R_I}\rightarrow \infty$.
\end{proposition}
\begin{proof}
The derivations of the above expressions and the proof of Theorem 2
are very similar to those for LBT, and they can be found
in~\citep{RenZhaoTR0705}.
\end{proof}

Similarly, based on (22, 23), we can obtain numerical solutions to
the constrained optimization problem (9) for ELBT.
Fig.~\ref{fig:trans_throughput_vs_p} shows the maximal transport
throughput as a function of the traffic load $p$ obtained by
optimizing $r_I$. We observe from
Fig.~\ref{fig:trans_throughput_vs_p} that RTS/CTS handshake
signaling improves the performance of LBT when the traffic load is
low, but it degrades the performance of LBT with best-effort
delivery when the traffic load is high. This suggests that even when
the overhead associated with RTS/CTS is ignored, RTS/CTS may lead to
performance degradation due to the asymmetry of spectrum
opportunities. When the traffic load is high, LBT with best-effort
delivery gives the best transport throughput. This is consistent
with our previous observation obtained from
Fig.~\ref{fig:trans_throughput_p_01_LBT} that best-effort is a more
suitable application to be considered for overlaying with a primary
network with relatively high traffic load.

It can also be observed from Fig.~\ref{fig:trans_throughput_vs_p}
that the maximal transport throughput of the ``genie aided'' system
(in which all opportunities are correctly detected) provides an
upper bound for LBT and ELBT, which matches our intuition very well.
Since LBT allows the secondary user to access the channel within the
interference constraint even when it is not an opportunity, compared
with the ``genie aided'' system, the increase of transport
throughput of LBT brought by accessing the \emph{busy} channel can
partially compensate the loss of transport throughput due to missed
detection of opportunities. When the traffic load of the primary
network is high, spectrum opportunities occur infrequently, and the
compensated transport throughput of LBT is comparable to the
transport throughput of the ``genie aided'' system. Thus, we see in
Fig.~\ref{fig:trans_throughput_vs_p} that the maximal transport
throughput of LBT for best-effort delivery approaches that of the
``genie aided'' system in this case. \vspace{-1em}

\section{Conclusion and Discussion}
\label{sec:conclu} \vspace{-0.5em}

We have studied transmission power control of secondary users in CR
networks. By carefully examining the concepts of spectrum
opportunity and interference constraint, we have revealed and
analytically characterized the impact of transmission power on the
occurrence of spectrum opportunities and the reliability of
opportunity detection. Based on a Poisson model of the primary
network, we have quantified these impacts by showing the exponential
decay rate of the probability of opportunity with respect to the
transmission power and the asymptotic behavior of the ROC for
opportunity detection. In the analysis, the independent coloring
theorem and displacement theorem have played a significant role,
especially the former one. Such analytical characterizations allow
us to design the transmission power for optimal transport throughput
under constraints on the interference to primary users.

Furthermore, the non-equivalence between detecting primary signals
and detecting spectrum opportunities has been illuminated. It has
been specified that besides noise and fading, the geographical
distribution and traffic pattern of primary users have significant
impact on the performance of physical layer spectrum sensing. The
complex dependency of the relationship between PHY and MAC on the
application types and the use of MAC handshake signaling such as
RTS/CTS is also illustrated.

In the above analysis, the interference region of primary users is
represented by a circle with radius $R_I$. It is possible that the
interference contributions from multiple interferers outside of this
circle cause an outage at the secondary receiver $B$, but by
choosing a conservative interference range $R_I$, this possibility
is negligible~\citep{Jindal&Etal:08ITWC}. To verify the validity of
the interference range $R_I$, we take into account the accumulated
interference power from all primary transmitters in the simulation,
which directly determines the Signal-to-interference ratio (SIR) at
the secondary receiver $B$. In this case, a channel is an
opportunity for one pair of secondary users if there is no primary
receiver within the interference range $r_I$ of the secondary
transmitter $A$ and the total power of the interference $I_B$ from
all primary transmitters to the secondary receiver $B$ is below some
prescribed level $\tau_B$, $\ie$
$\overline{\mathbb{I}(A,r_I,\textrm{rx})}\cap (I_B<\tau_B)$. To
detect spectrum opportunity, the secondary transmitter $A$ uses an
energy detector, which is given by
\begin{eqnarray*}
I_A\overset{\mathcal{H}_0}{\underset{\mathcal{H}_1}{\lesseqgtr}}
\tau_A,
\end{eqnarray*}
where $I_A$ is the total received power at the secondary transmitter
$A$ and $\tau_A$ is the threshold of the energy detector. Let
$\Pi_{tx}$ denote the Poisson point process of primary transmitters,
$r^A_i$ and $r^B_i$ the distance from the $i^{\textrm{th}}$ primary
transmitter to $A$ and $B$ respectively, and $\alpha$ the path-loss
exponent, then we have that
\begin{eqnarray*}
I_A = \sum_{i\in \Pi_{tx}} P_{tx}\cdot (r^A_i)^{-\alpha},~~~~~~I_B =
\sum_{i\in \Pi_{tx}} P_{tx}\cdot (r^B_i)^{-\alpha}.
\end{eqnarray*}

Fig.~\ref{fig:ROC_p_001_simu} shows the simulated ROC of the energy
detector and the analytical ROC of LBT. We can see that reliable
opportunity detection can be achieved when
$\frac{p_{tx}}{P_{tx}}\rightarrow 0$ and
$\frac{p_{tx}}{P_{tx}}\rightarrow\infty$. In other words, the
asymptotic property of ROC (Theorem 2) still holds in this case.

As an initial analysis, we focus on one pair of secondary users in
this paper. In the scenario of multiple secondary users, our
definition of spectrum opportunity can still be applied, although
determining the interference range $r_I$ of secondary users needs
careful consideration due to the accumulation of the interference
powers from multiple secondary transmitters. Moreover, since the
reception at one secondary user can be affected by other secondary
transmissions, the detection of opportunities should probably be
performed cooperatively to avoid the possibly violent contention of
accessing the idle channel among secondary users. Our results can
easily be extended to incorporate the trade-offs between sampling.
vs opportunity to transmit by incorporation in a spatial Poisson
model. To specify the effect of multiple secondary users on power
control is an interesting future direction, and some preliminary
results about the impact of transmission power of secondary users on
the connectivity of the secondary network can be found
in~\citep{Ren&Etal:09MobiHoc}. \vspace{-1em}

{\small
\section*{Appendix A: Simplification of Double Integral for
$\Pr[\mathcal{H}_0]$} \vspace{-0.5em}

\renewcommand{\theequation}{A\arabic{equation}}
\setcounter{equation}{0}

By considering the shape of $\Sc_c(d,r_I+R_p,R_I)$
(see~\citep{RenZhaoTR0705}), we can use the independence of the
integrand on the angular coordinate $\theta$ to reduce the double
integral to a single integral with respect to the radial coordinate
$r$. Here we choose the secondary transmitter $A$ as the origin of
the polar coordinate system and the line from the secondary
transmitter $A$ to the secondary receiver $B$ as the polar axis. Due
to the symmetry of $\Sc_c(d,r_I+R_p,R_I)$ with respect to the polar
axis, there is a coefficient $2$ before each integral below. Let
$Q=\underset{\Sc_c(d,r_I+R_p,R_I)}{\iint}\frac{\Sc_I(r,R_p,r_I)}{\pi
R_p^2}r\mathrm{d}r\mathrm{d}\theta$, $I(r,R_p,r_I)=2\pi \int_{0}^{r}
t\frac{\Sc_I(t,R_p,r_I)}{\pi R_p^2}\mathrm{d}t$,
$\theta_0(r)=\arccos\left(\frac{d^2+r^2-R_I^2}{2dr}\right)$, then we
have two cases:
\begin{enumerate}
\item[$\Box$] Case 1: $R_I\geq d$.
\begin{itemize}
\item If $r_I+R_p\leq R_I-d$, then $Q=0$.
\item If $R_I-d<r_I+R_p<R_I+d$, then $Q=\pi r_I^2-I(R_I-d,R_p,r_I)-2\int_{R_I-d}^{r_I+R_p}\frac{\Sc_I(r,R_p,r_I)}{\pi
R_p^2}r \theta_0(r)\mathrm{d}r$.
\item If $r_I+R_p>R_I+d$, then $Q=\pi r_I^2-I(R_I-d,R_p,r_I)-2\int_{R_I-d}^{R_I+d}\frac{\Sc_I(r,R_p,r_I)}{\pi
R_p^2}r \theta_0(r)\mathrm{d}r$.
\end{itemize}

\item[$\Box$] Case 2: $R_I<d$.
\begin{itemize}
\item If $r_I+R_p\leq d-R_I$, then
$Q=\pi r_I^2$.
\item If $d-R_I<r_I+R_p<d+R_I$, then $Q=\pi r_I^2-2\int_{d-R_I}^{r_I+R_p}\frac{\Sc_I(r,R_p,r_I)}{\pi
R_p^2}r \theta_0(r)\mathrm{d}r$.
\item If $r_I+R_p>d+R_I$, then $Q=\pi r_I^2-2\int_{d-R_I}^{d+R_I}\frac{\Sc_I(r,R_p,r_I)}{\pi R_p^2}r
\theta_0(r)\mathrm{d}r$.
\end{itemize}
\end{enumerate}

The expression for $I(r,R_p,r_I)$ is obtained in an explicit form as
listed below.
\begin{enumerate}
\item[$\Box$] Case 1: for $r\leq |R_p-r_I|$, $I(r,R_p,r_I)=\pi r^2 \min
\{1,\frac{r_I^2}{R_p^2}\}$.
\item[$\Box$] Case 2: for $|R_p-r_I|<r<R_p+r_I$, then $I(r,R_p,r_I)=\frac{1}{2}\pi r_I^2+r^2\arccos\left(\frac{R_p^2+r^2-r_I^2}{2R_p r}\right)
+\frac{r_I^2 r^2}{R_p^2}\arccos\left(\frac{r_I^2+r^2-R_p^2}{2r_I
r}\right)-r_I^2\arcsin\left(\frac{r_I^2+R_p^2-r^2}{2r_I
R_p}\right)-\frac{r^2+r_I^2+R_p^2}{4
R_p^2}\sqrt{(r_I+R_p+r)(r_I+R_p-r)(r_I-R_p+r)(R_p-r_I+r)}$.
\item[$\Box$] Case 3: for $r\geq R_p+r_I$, then $I(r,R_p,r_I)=\pi r_I^2$.
\end{enumerate}

To compute the remaining integral $\int \frac{\Sc_I(r,R_p,r_I)}{\pi
R_p^2}r \theta_0(r)\mathrm{d}r$ numerically, we need an
explicit-form expression for $\Sc_I(r,R_p,r_I)$. Let
$r_1=\min\{R_p,r_I\}$ and $r_2=\max\{R_p,r_I\}$, then the expression
for $\Sc_I(r,R_p,r_I)$ is given by
\begin{enumerate}
\item[$\Box$] Case 1: for $0\leq r\leq r_2-r_1$, $\Sc_I(r,R_p,r_I)=\pi r_1^2$.
\item[$\Box$] Case 2~\citep{Two_circles}: for $r_2-r_1<r<
r_1+r_2$, \vspace{-0.5em}
\begin{eqnarray*}
\Sc_I(r,R_p,r_I)=r_2^2\cos^{-1}\left(\frac{r_2^2+r^2-r_1^2}{2r_2 r}\right)+r_1^2\cos^{-1}\left(\frac{r_1^2+r^2-r_2^2}{2r_1 r}\right)%
-r\sqrt{r_1^2-\left(\frac{r_1^2+r^2-r_2^2}{2r}\right)^2}.
\end{eqnarray*}
\item[$\Box$] Case 3 for $r \geq r_1+r_2$, $\Sc_I(r,R_p,r_I)=0$.
\end{enumerate}
\vspace{-1em}

\section*{Appendix B: Proof of Theorem 1}
\vspace{-0.5em}

\renewcommand{\theequation}{B\arabic{equation}}
\setcounter{equation}{0}

Since the integrand $\frac{\Sc_I(r,R_p,r_I)}{\pi R_p^2}$ and the
region of the double integral $\Sc_c(d,r_I+R_p,R_I)$ are both
increasing functions of $r_I$~\citep{RenZhaoTR0705}, T1.1 follows
from the monotonicity of the exponential function. Obviously, T1.3
also follows from the monotonicity of the exponential function.

We now prove T1.2. Recall the definition of spectrum opportunity,
\vspace{-0.5em}
\begin{eqnarray}
\Pr[\mathcal{H}_0] =
\Pr\{\overline{\mathbb{I}(A,r_I,\textrm{rx})}\cap
\overline{\mathbb{I}(B,R_I,\textrm{tx})}\}
= \Pr\{\overline{\mathbb{I}(A,r_I,\textrm{rx})}~|~\overline{\mathbb{I}(B,R_I,\textrm{tx})}\}\cdot \Pr\{\overline{\mathbb{I}(B,R_I,\textrm{tx})}\}.%
\end{eqnarray}
Based on Property 1, we have that for all $r_I>0$, \vspace{-0.5em}
\begin{eqnarray}
\Pr\{\overline{\mathbb{I}(B,R_I,\textrm{tx})}\} &=& \exp(-p\lambda
\pi
R_I^2), \\
\Pr\{\overline{\mathbb{I}(A,r_I,\textrm{rx})}~|~\overline{\mathbb{I}(B,R_I,\textrm{tx})}\}
&>&\Pr\{\overline{\mathbb{I}(A,r_I,\textrm{rx})}\}=\exp(-p\lambda
\pi r_I^2).
\end{eqnarray}
In the last inequality, we have used the fact that the logical
condition $\overline{\mathbb{I}(B,R_I,\textrm{tx})}$ reduces the
number of primary transmitters that can communicate with primary
users within distance $r_I$ to $A$, which results in a more probable
occurrence of $\overline{\mathbb{I}(A,r_I,\textrm{rx})}$. Then by
substituting (B2, B3) into (B1), we have
$\Pr[\mathcal{H}_0]>\exp[-p\lambda \pi(r_I^2+R_I^2)]$ for all $r_I$.

Obviously, $\Pr[\mathcal{H}_0]\leq \exp(-p\lambda \pi r_I^2)$.
Moreover, when $r_I\geq d+R_I+R_p$, we have
(see~\citep{RenZhaoTR0705}) \vspace{-0.5em}
\begin{eqnarray*}
\underset{\Sc_c(d,r_I+R_p,R_I)}{\iint}\frac{\Sc_I(r,R_p,r_I)}{\pi
R_p^2}r\mathrm{d}r\mathrm{d}\theta = \pi(r_I^2-R_I^2).
\end{eqnarray*}
So $\Pr[\mathcal{H}_0] = \exp(-p\lambda \pi r_I^2)$ when $r_I\geq
d+R_I+R_p$. \vspace{-1em}

\section*{Appendix C: Proof of Proposition 2}
\vspace{-0.5em}

\renewcommand{\theequation}{C\arabic{equation}}
\setcounter{equation}{0}

From (1), we have \vspace{-0.5em}
\begin{eqnarray}
P_F &=& \Pr\{\mathbb{I}(A,r_D,\textrm{tx})~|~\mathcal{H}_0\} =
1-\Pr\{\overline{\mathbb{I}(A,r_D,\textrm{tx})}~|~\mathcal{H}_0\},
\\%
P_{MD} &=&
\Pr\{\overline{\mathbb{I}(A,r_D,\textrm{tx})}~|~\mathcal{H}_1\}
= \Pr\{\overline{\mathbb{I}(A,r_D,\textrm{tx})}~|~\mathbb{I}(A,r_I,\textrm{rx})\cup \mathbb{I}(B,R_I,\textrm{tx})\}. %
\end{eqnarray}

Based on the definition of spectrum opportunity, we have %
\begin{eqnarray}
P_F &=& 1-\frac{\Pr\{\overline{\mathbb{I}(A,r_D,\textrm{tx})}\cap
\overline{\mathbb{I}(A,r_I,\textrm{rx})}\cap
\overline{\mathbb{I}(B,R_I,\textrm{tx})}\}}{\Pr[\mathcal{H}_0]}, \\%
P_{MD}&=&
\frac{\Pr\{\overline{\mathbb{I}(A,r_D,\textrm{tx})}\}-\Pr\{\overline{\mathbb{I}(A,r_D,\textrm{tx})}\cap
\overline{\mathbb{I}(A,r_I,\textrm{rx})}\cap
\overline{\mathbb{I}(B,R_I,\textrm{tx})}\}}{1-\Pr[\mathcal{H}_0]}.
\end{eqnarray}

Since $\Pr[\mathcal{H}_0]$ is known, we only need to calculate the two probabilities in the numerators. %

Based on Property 1, we have \vspace{-0.5em}
\begin{eqnarray}
& &\Pr\{\overline{\mathbb{I}(A,r_D,\textrm{tx})}\}=\exp(-p\lambda
\pi r_D^2), \\
& &\Pr\{\overline{\mathbb{I}(A,r_D,\textrm{tx})}\cap \overline{\mathbb{I}(B,R_I,\textrm{tx})}\} = \exp[-p\lambda(\pi(r_D^2+R_I^2)-\Sc_I(d,r_D,R_I))]. %
\end{eqnarray}

By using techniques similar to those used in obtaining
$\Pr\{\overline{\mathbb{I}(A,r_I,\textrm{rx})}~|~\overline{\mathbb{I}(B,R_I,\textrm{tx})}\}$
(see the proof of Proposition 1 in Sec.~\ref{subsec:Imp_Tx_Pr_H0}),
we have \vspace{-0.5em}
\begin{eqnarray}
\Pr\{\overline{\mathbb{I}(A,r_I,\textrm{rx})}~|~\overline{\mathbb{I}(A,r_D,\textrm{tx})}\cap
\overline{\mathbb{I}(B,R_I,\textrm{tx})}\}=\exp\left(-\underset{\Sc_c(d,r_I+R_p,R_I)-\Sc_{A2}}
{\iint}p\lambda \frac{\Sc_I(r,R_p,r_I)}{\pi
R_p^2}r\mathrm{d}r\mathrm{d}\theta\right).~~
\end{eqnarray}
Since $\Pr\{\overline{\mathbb{I}(A,r_D,\textrm{tx})}\cap
\overline{\mathbb{I}(A,r_I,\textrm{rx})}\cap
\overline{\mathbb{I}(B,R_I,\textrm{tx})}\}
=\Pr\{\overline{\mathbb{I}(A,r_I,\textrm{rx})}~|~\overline{\mathbb{I}(A,r_D,\textrm{tx})}\cap \overline{\mathbb{I}(B,R_I,\textrm{tx})}\} %
   \Pr\{\overline{\mathbb{I}(A,r_D,\textrm{tx})}\cap
   \overline{\mathbb{I}(B,R_I,\textrm{tx})}\}$,
by substituting (2, C5-C7) into (C3, C4), we obtain (6), (7).
\vspace{-1em}

\section*{Appendix D: Proof of Theorem 2}
\vspace{-0.5em}

\renewcommand{\theequation}{D\arabic{equation}}
\setcounter{equation}{0}

Consider first $\frac{r_I}{R_I}\rightarrow 0$. As discussed in
Sec.~\ref{subsec:Det_Oppor}, we choose $r_D=R_I$ in this case.
Recall that $\Sc_o=\Sc_c(d,r_D,R_I)\cap \Sc_c(d,r_I+R_p,R_I)$ as
given in Proposition 2, and then we have \vspace{-0.5em}
\begin{eqnarray*}
\underset{\frac{r_I}{R_I}\rightarrow
0}{\lim}|\Sc_c(d,r_I+R_p,R_I)|=0,~~~~
\underset{\frac{r_I}{R_I}\rightarrow 0}{\lim}|\Sc_o|=0,~~~~0\leq
\underset{\frac{r_I}{R_I}\rightarrow
0}{\lim}\underset{\Sc_o}{\iint}\frac{\Sc_I(r,R_p,r_I)}{\pi
R_p^2}r\mathrm{d}r\mathrm{d}\theta \leq
\underset{\frac{r_I}{R_I}\rightarrow
0}{\lim}\underset{\Sc_o}{\iint}\frac{r_I^2}{R_p^2}r\mathrm{d}r\mathrm{d}\theta=0.
\end{eqnarray*}
So by substituting the above limits into (6, 7) and applying the
continuity of functions in (6, 7), we conclude that
$P_F(r_D=R_I)\rightarrow 0$, $P_{MD}(r_D=R_I)\rightarrow 0$ as
$r_I/R_I\rightarrow 0$.

Consider next $\frac{r_I}{R_I}\rightarrow \infty$. As discussed in
Sec.~\ref{subsec:Det_Oppor}, we choose $r_D=r_I-R_I$ in this case.
Then we have \vspace{-0.5em}
\begin{eqnarray*}
\underset{\frac{r_I}{R_I}\rightarrow
\infty}{\lim}\underset{\Sc_c(d,r_I+R_p,R_I)}{\iint}\frac{\Sc_I(r,R_p,r_I)}{\pi
R_p^2}r\mathrm{d}r\mathrm{d}\theta = \pi (r_I^2-R_I^2),~~~~
\underset{\frac{r_I}{R_I}\rightarrow
\infty}{\lim}\underset{\Sc_o}{\iint}\frac{\Sc_I(r,R_p,r_I)}{\pi
R_p^2}r\mathrm{d}r\mathrm{d}\theta = \pi
\left[(r_I-R_I)^2-R_I^2\right].
\end{eqnarray*}
Similarly, we can show that $P_F(r_D=r_I-R_I)\rightarrow 0$,
$P_{MD}(r_D=r_I-R_I)\rightarrow 0$ as $r_I/R_I\rightarrow \infty$.
\vspace{-1em}

\section*{Appendix E: Derivation of Collision Probability $P_C$ in Proposition 3}
\vspace{-0.5em}

\renewcommand{\theequation}{E\arabic{equation}}
\setcounter{equation}{0}

Recall (13) and use the total probability theorem to obtain,
\begin{eqnarray*}
P_C &=& \frac{\Pr\{\overline{\mathbb{I}(A,r_D,\textrm{tx})}\cap
\mathbb{I}(A,r_I,\textrm{rx})\}}{\Pr\{\mathbb{I}(A,r_I,\textrm{rx})\}},
\\ &=&
\frac{\Pr\{\overline{\mathbb{I}(A,r_D,\textrm{tx})}\}-\Pr\{\overline{\mathbb{I}(A,r_I,\textrm{rx})}~|~\overline{\mathbb{I}(A,r_D,\textrm{tx})}\}\cdot
\Pr\{\overline{\mathbb{I}(A,r_D,\textrm{tx})}\}}{\Pr\{\mathbb{I}(A,r_I,\textrm{rx})\}}.
\end{eqnarray*}

It follows from Property 1 that
$\Pr\{\mathbb{I}(A,r_I,\textrm{rx})\}=1-\exp(-p\lambda \pi r_I^2)$,
$\Pr\{\overline{\mathbb{I}(A,r_D,\textrm{tx})}\}=\exp(-p\lambda \pi
r_D^2)$. Then by using arguments similar to those used in obtaining
$\Pr\{\overline{\mathbb{I}(A,r_I,\textrm{rx})}~|~\overline{\mathbb{I}(B,R_I,\textrm{tx})}\}$
(see the proof of Proposition 1 in Sec.~\ref{subsec:Imp_Tx_Pr_H0}),
we obtain the expression for
$\Pr\{\overline{\mathbb{I}(A,r_I,\textrm{rx})}~|~\overline{\mathbb{I}(A,r_D,\textrm{tx})}\}$,
and (14) follows immediately. Notice that from (14) we can show that
$P_C$ decreases as $r_D$ and $p\lambda$
increases~\citep{RenZhaoTR0705}. It follows that for fixed $r_I$,
$r_D$ decreases as $p\lambda$ increases in order to satisfy the
collision constraint. \vspace{-1em}

\bibliographystyle{ieeetr}

}

\newpage

\begin{figure}[htb]
\centerline{
\begin{psfrags}
\psfrag{A}[c]{ $A$} \psfrag{B}[c]{ $B~$} \psfrag{i}[c]{
~~~~~~~~~~~Interference} \psfrag{ra}[r]{ {$r_I$}} \psfrag{rb}[c]{
{$~~R_I$}} \psfrag{tx}[l]{ Primary Tx} \psfrag{rx}[l]{ Primary Rx}
\scalefig{0.5}\epsfbox{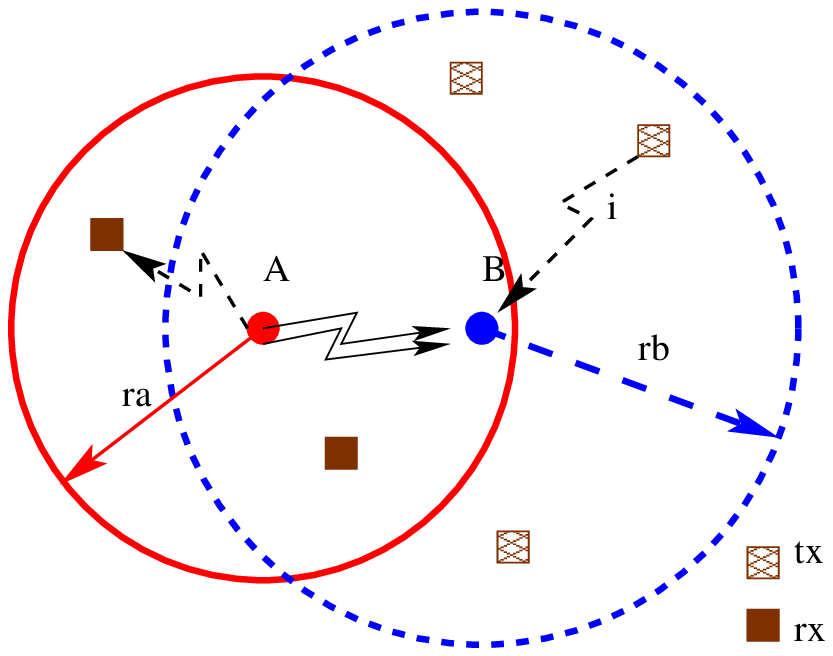}
\end{psfrags}}
\caption{Illustration of spectrum opportunity (secondary user $A$
wishes to transmit to secondary user $B$, where $A$ should watch for
nearby primary receivers and $B$ nearby primary transmitters).}
\label{fig:SO}
\end{figure}

\begin{figure}[htbp]
\centerline{
\begin{psfrags}
\psfrag{A}[c]{$A$} \psfrag{B}[c]{$B$} \psfrag{X}[c]{$X$}
\psfrag{Y}[c]{$Y$} \psfrag{Z}[c]{$Z$} \psfrag{tx}[l]{ Primary Tx}
\psfrag{rx}[l]{ Primary Rx} \psfrag{rI}[c]{$r_I$}
\psfrag{RI}[c]{$R_I$} \psfrag{rd}[c]{$r_D$} \psfrag{d}[c]{$d$}
\psfrag{Interference}[c]{\small{Interference}}
\scalefig{0.45}\epsfbox{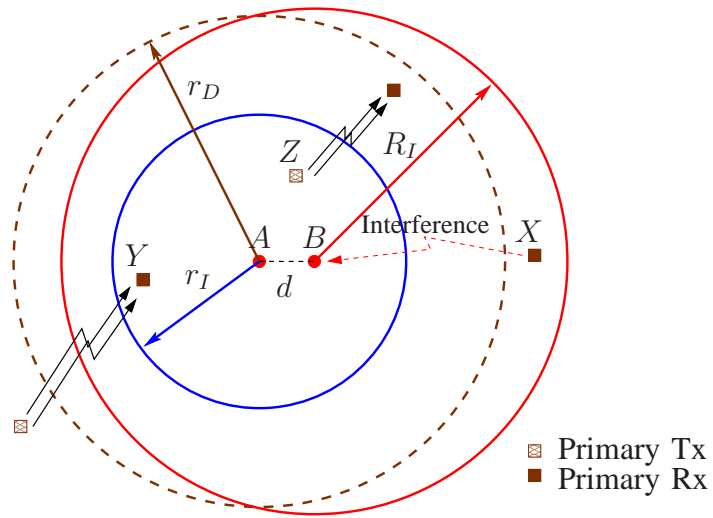}
\end{psfrags}
} \caption{Spectrum opportunity detection: A common approach that
detects spectrum opportunities by observing primary signals (the
exposed transmitter $Z$ is a source of false alarms whereas the
hidden transmitter $X$ and the hidden receiver $Y$ are sources of
miss detections).} \label{fig:OD}
\end{figure}

\begin{figure}[htbp]
\centerline{
\begin{psfrags}
\psfrag{d}[l]{ $P_D=1-P_{MD}$} \psfrag{e}[c]{ $P_F$} \psfrag{o}[c]{
$0$} \psfrag{a}[c]{ $1$} \psfrag{r1}[l]{ $r_D \downarrow$}
\psfrag{r2}[l]{ {$r_{D} \uparrow$}}
\scalefig{0.45}\epsfbox{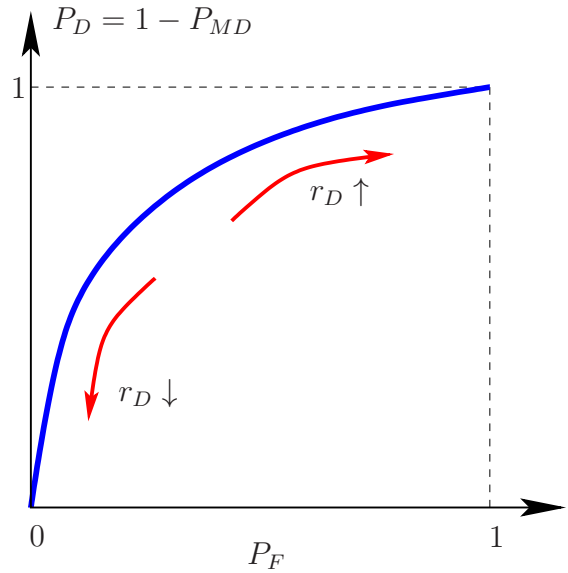}
\end{psfrags}
} \caption{ROC of LBT with perfect ears (the ROC is obtained by
varying the detection range $r_D$).} \label{fig:ROC}
\end{figure}

\begin{figure}[htbp]
\centerline { \hspace{0.2in}
\begin{psfrags}
\psfrag{r1}[c]{$r_1$} \psfrag{r2}[c]{$r_2$} %
\psfrag{A}[c]{$A$} \psfrag{B}[c]{$B$} \psfrag{d}[l]{$d$} %
\psfrag{Sc}[c]{$\Sc_c(d,r_1,r_2)$}
\psfrag{SI}[c]{$\Sc_I(d,r_1,r_2)$}
\scalefig{0.5}\epsfbox{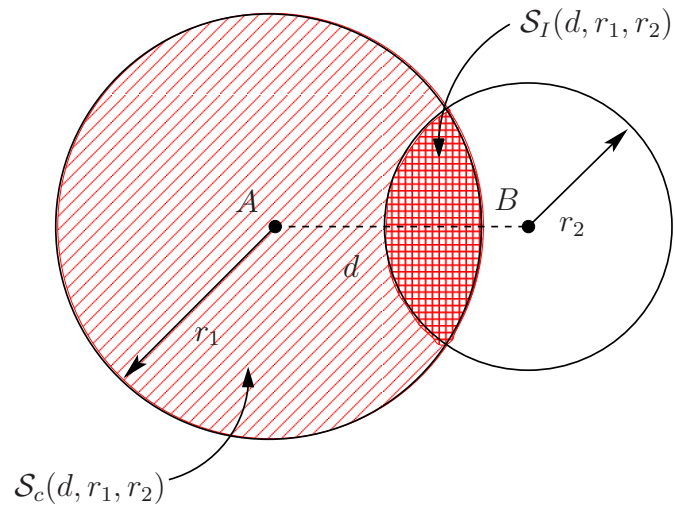}
\end{psfrags}
} \caption{Illustration of $\Sc_I(d,r_1,r_2)$ (the common area of
two circles with radius $r_1$ and $r_2$ and centered $d$ apart) and
$\Sc_c(d,r_1,r_2)$ (the area within a circle with radius $r_1$
centered at $A$ but outside the circle with radius $r_2$ centered at
$B$ which is distance $d$ away from $A$).} \label{fig:Sc_0}
\end{figure}

\begin{figure}
\begin{minipage}{3in}
\centerline{\scalefig{1.15} \epsfbox{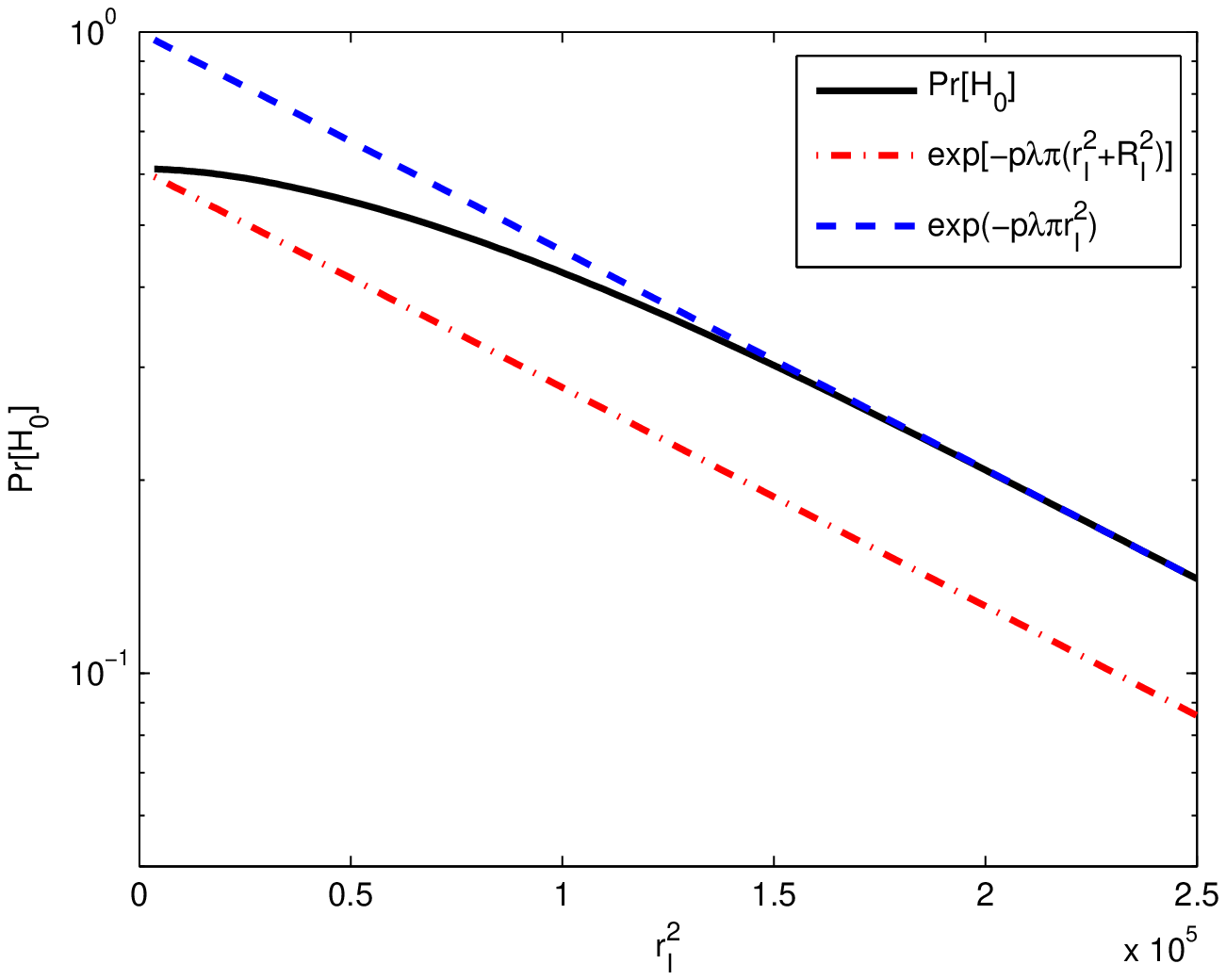}}
\centerline{(a)} %
\end{minipage}
~ %
\begin{minipage}{3in}
\centerline{\scalefig{1.15} \epsfbox{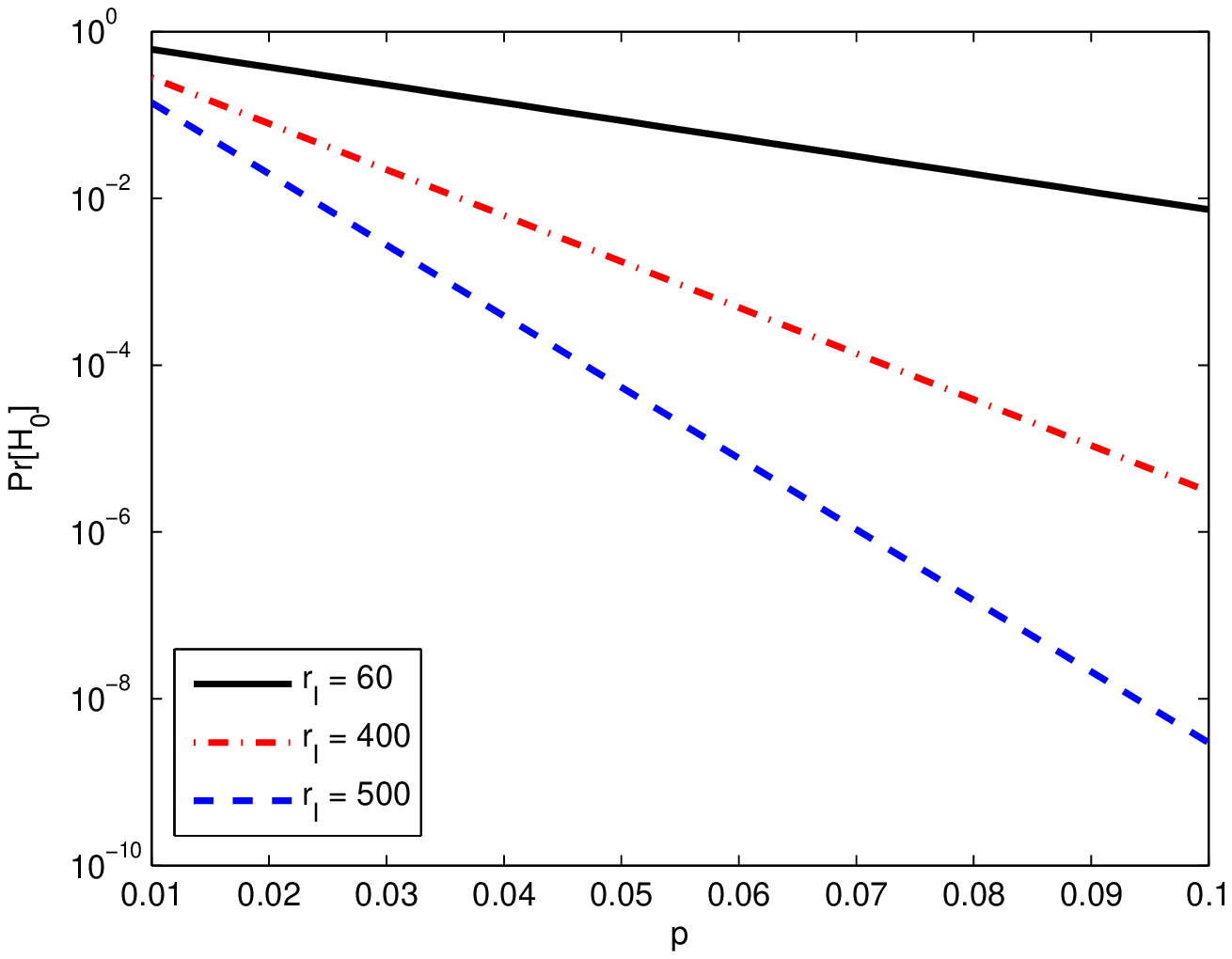}} %
\centerline{(b)}
\end{minipage}
\caption{(a) $\Pr[\mathcal{H}_0]$ vs $r_I$ ($p=0.01$,
$\lambda=10/200^2$, $d=50$, $R_p=200$, $R_I=250$); (b)
$\Pr[\mathcal{H}_0]$ vs $p$ ($\lambda=10/200^2$, $d=50$, $R_p=200$,
$R_I=250$). Note that both y-axes use log-scale.}
\label{fig:log_Pr_H0}
\end{figure}

\begin{figure}[htbp]
\centerline{\scalefig{0.7} \epsfbox{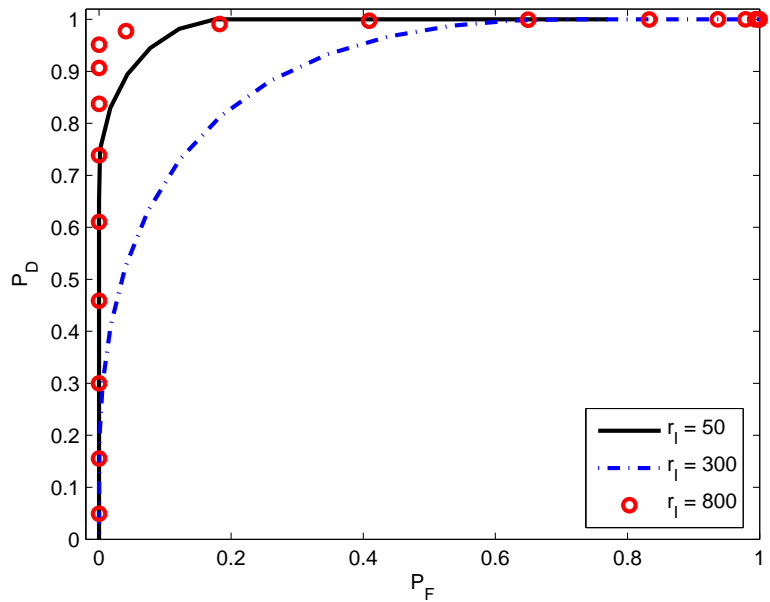}} %
\caption{ROC for LBT ($p=0.01$, $\lambda = 10/200^2$, $R_p=200$,
$R_I=R_p/0.8=250$, $d=0.9r_I$)} \label{fig:ROC_extreme_LBT}
\end{figure}

\begin{figure}[htbp]
\centerline{\scalefig{0.7} \epsfbox{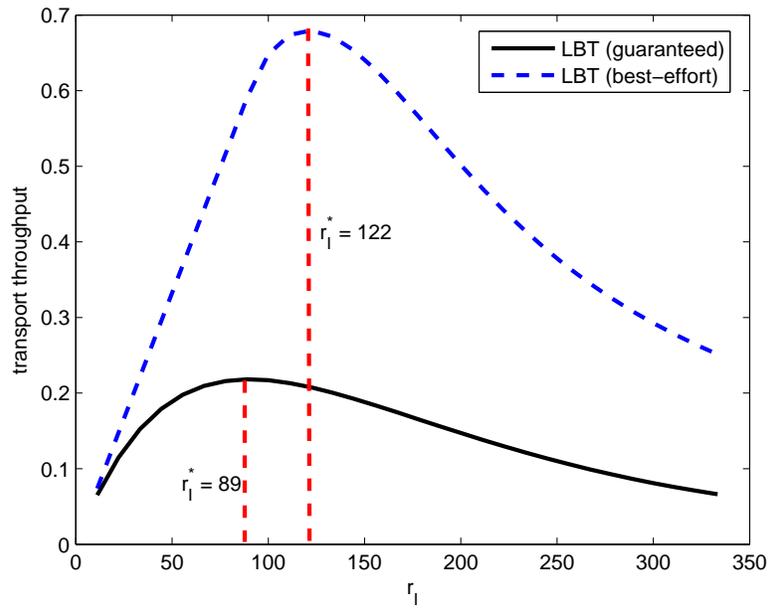}} %
\caption{Transport Throughput vs $r_I$ ($p=0.1$, $\lambda =
10/200^2$, $R_p=200$, $R_I=R_p/0.8=250$, $d=0.9r_I$, $\zeta=0.05$)}
\label{fig:trans_throughput_p_01_LBT}
\end{figure}

\begin{figure}[htbp]
\centerline{\scalefig{0.7} \epsfbox{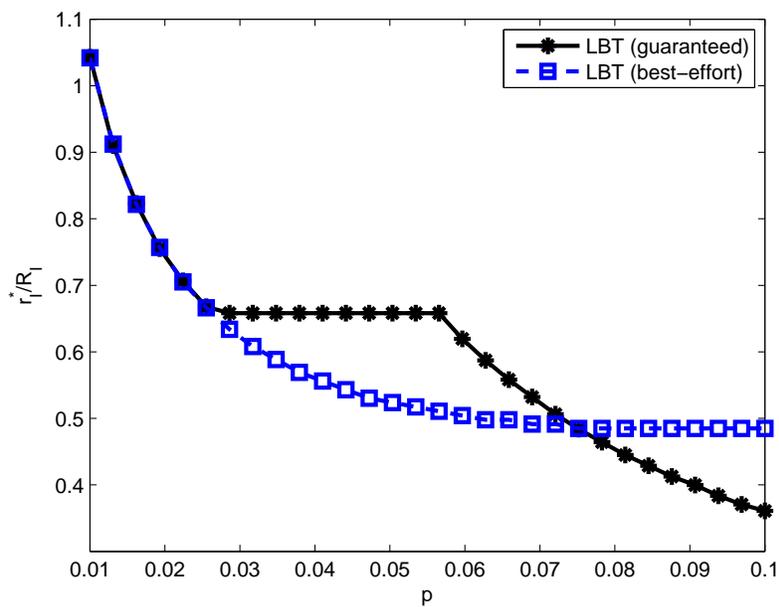}} %
\caption{Ratio between optimal interference range $r_I^*$ for
transport throughput and $R_I$ vs traffic load $p$ of primary users
($\lambda=10/200^2$, $R_p=200$, $R_I=R_p/0.8=250$, $d=0.9r_I$,
$\zeta=0.05$)} \label{fig:opt_r_I_LBT}
\end{figure}

\begin{figure}[htbp]
\centerline{\scalefig{0.7} \epsfbox{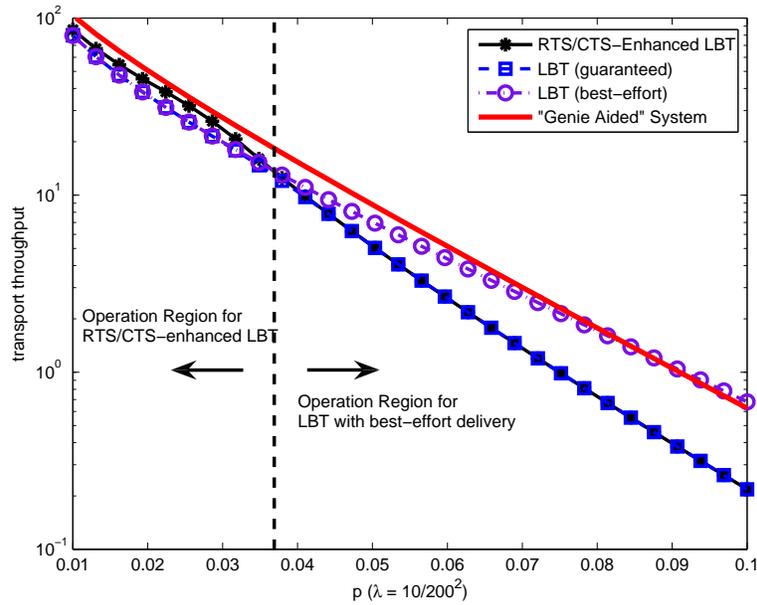}} %
\caption{Transport throughput vs traffic load $p$ of primary users
($\lambda=10/200^2$, $R_p=200$, $R_I=R_p/0.8=250$, $d=0.9r_I$,
$\zeta=0.05$)} \label{fig:trans_throughput_vs_p}
\end{figure}

\begin{figure}[htbp]
\centerline{\scalefig{0.7} \epsfbox{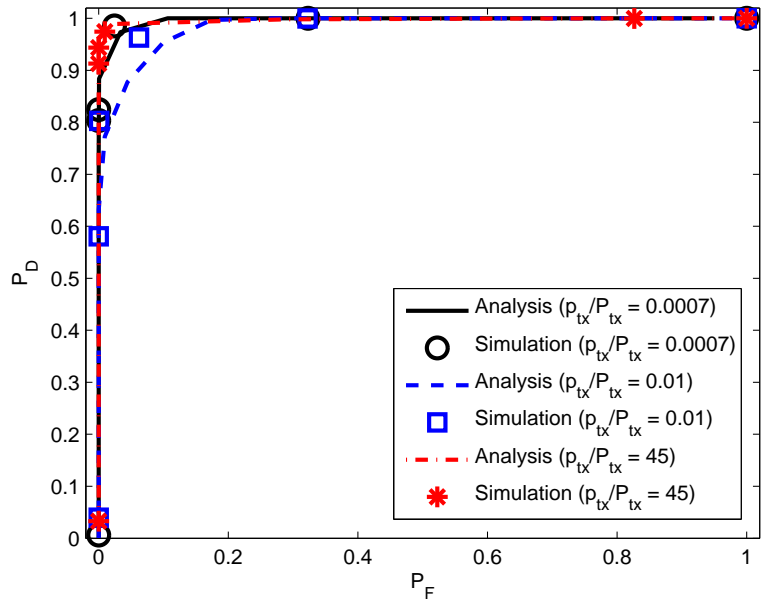}} %
\caption{Simulated ROC of energy detector vs. analytical ROC of LBT
($p=0.01$, $\lambda=10/200^2$, $P_{tx}=10$, $R_p=200$,
$R_I=R_p/0.8=250$, $\tau_B=P_{tx}\cdot R_I^{-\alpha}$, $\alpha = 3$,
$r_I=(p_{tx}/\tau_B)^{\frac{1}{\alpha}}$, $d=0.9\cdot r_I$)}
\label{fig:ROC_p_001_simu}
\end{figure}

\end{document}